\documentclass[
12pt,
prx,
reprint,
a4paper,
longbibliograph,
nofootinbib,
nobibnotes,
superscriptaddress
]{revtex4-2}
\usepackage[utf8]{inputenc}
\usepackage[T1]{fontenc}

\usepackage[english]{babel}

\usepackage[unicode=true]{hyperref}
\hypersetup{
     colorlinks=true,       		% false: boxed links; true: colored links
     linkcolor=blue,          	% color of internal links
     citecolor=blue,            % color of links to bibliography
     urlcolor=blue,           	% color of external links
}
\usepackage{amsmath,amssymb,amsthm,mathtools}
\usepackage[dvipsnames]{xcolor}
\usepackage{tikz}
\usetikzlibrary{matrix,arrows.meta}
\usetikzlibrary{positioning,fit,calc}
\tikzstyle{block} = [draw=black, thick, text width=2cm, minimum height=1cm, align=center]  
\tikzstyle{arrow} = [thick,->,>=stealth]
\usepackage{textcmds}
\usepackage{soul}

\definecolor{dblue}{RGB}{59,104,118}
\definecolor{lgreen}{RGB}{150,170,50}

\hypersetup{
     colorlinks=true,       		% false: boxed links; true: colored links
     linkcolor=dblue,          	% color of internal links
     citecolor=dblue,            % color of links to bibliography
     urlcolor=lgreen,           	% color of external links
}

\theoremstyle{definition}
\newtheorem{theorem}{Theorem}
\newtheorem{example}[theorem]{Example}

\newtheorem{corollary}[theorem]{Corollary}

\newtheorem{lemma}[theorem]{Lemma}
\newtheorem{definition}[theorem]{Definition}
\newtheorem{proposition}[theorem]{Proposition}

\DeclareMathOperator\tr{Tr}
\newcommand\mean[1]{\left\langle #1 \right\rangle}
\def\id{\mathbb{I}}%
\def\sos{\mathit{q}}%
\def\sos#1{\mathit{q}  \left(#1\right)   }%
\def\SOS{\mathcal{Q}}%
\def\SOS#1{\mathcal{Q}\left( #1 \right)}%
\def\cSOS#1{\overline{\mathcal{Q}}\left( #1 \right)}%
\def\JNR#1{\mathcal{J}\left( #1 \right)}%
\def\stab{\operatorname{STAB}}
\def\th{\operatorname{TH}}

\def\saur{\mathcal{S}_{ac}}%
\def\sar{\mathcal{S}^{{\scriptscriptstyle\leq}}_{ac}}%
\def\saura{\mathcal{S}_a}%
\def\sara{\mathcal{S}^{{\scriptscriptstyle\leq}}_a}%

\def\eqabove#1  {\stackrel{\mathclap{\footnotesize\mbox{#1} }} {=}     }%
\def\leqabove#1 {\stackrel{\mathclap{\normalfont\mbox{#1} }} {\leq}  }%
\def\EXP#1{\langle #1 \rangle_\rho }%

\begin{document}

\title{Bounding the joint numerical range of Pauli strings by graph parameters}

\author{Zhen-Peng Xu}
\email{zhen-peng.xu@ahu.edu.cn}
\affiliation{School of Physics and Optoelectronics Engineering,\protect\\ Anhui University, 230601 Hefei, People’s Republic of China}
\affiliation{Naturwissenschaftlich--Technische Fakult\"{a}t,
Universit\"{a}t Siegen, 57068 Siegen, Germany}

\author{Ren\'{e} Schwonnek}
\email{rene.schwonnek@itp.uni-hannover.de}
\affiliation{Institut f\"{u}r Theoretische Physik, Leibniz Universit\"{a}t Hannover, Appelstrasse 2, 30167, Germany}

\author{Andreas Winter}
\email{andreas.winter@uab.cat}
%\affiliation{ICREA---Instituci\'o Catalana de Recerca i Estudis Avan\c{c}ats,\protect\\ Pg.~Lluis Companys, 23, 08010 Barcelona, Spain} 
\affiliation{ICREA {\&} Grup d'Informaci\'{o} Qu\`{a}ntica, Departament de F\'{\i}sica,\protect\\ Universitat Aut\`{o}noma de Barcelona, 08193 Bellaterra (Barcelona), Spain}
\affiliation{Institute for Advanced Study, Technische Universit\"at M\"unchen,\protect\\ Lichtenbergstra{\ss}e 2a, 85748 Garching, Germany}
\affiliation{QUIRCK--Quantum Information Independent Research Centre Kessenich,\protect\\ Gerhard-Samuel-Stra{\ss}e 14, 53129 Bonn, Germany}

\begin{abstract}
% The interplay between the quantum state space and a specific set of measurements can be effectively captured by examining the set of jointly attainable expectation values. This set is commonly referred to as the (convex) joint numerical range. In this work, we explore geometric properties of this construct for measurements represented by tensor products of Pauli observables, also known as Pauli strings.
% The structure of pairwise commutation and anti-commutation relations among a set of Pauli strings determines a graph $G$, sometimes also called the frustration graph. We investigate the connection between the parameters of this graph and the structure of minimal ellipsoids encompassing the joint numerical range. Such an outer approximation can be very practical since ellipsoids can be handled analytically even in high dimensions. 

% We find counterexamples to a conjecture from [C. de Gois, K. Hansenne, O. G\"uhne, Phys. Rev. A 107, 062211 (2023)], and answer an open question in [M. B. Hastings and R. O’Donnell, {Proc.~STOC 2022, pp.~776–789}], which implies a new graph parameter that we call $\beta(G)$. Besides, we develop this approach in different directions, such as comparison with graph-theoretic approaches in other fields, applications in quantum information theory, numerical methods, properties of the new graph parameter, etc. Our approach suggests many open questions that we discuss briefly at the end.
%
The relations among a given set of observables on a quantum system are effectively captured by their so-called joint numerical range, which is the set of tuples of jointly attainable expectation values. Here, we explore geometric properties of this construct for Pauli strings, whose pairwise commutation and anti-commutation relations determine a graph $G$. We investigate the connection between the parameters of this graph and the structure of minimal ellipsoids encompassing the joint numerical range, and develop this approach in different directions. As a consequence, we find counterexamples to a conjecture by de Gois, Hansenne, and G\"uhne [Phys. Rev. A 107, 062211 (2023)], and answer an open question of Hastings and O'Donnell [Proc. STOC 2022, pp. 776-789], which implies a new graph parameter that we call $\beta(G)$. 
Furthermore, we provide new insights into the perennial problem of estimating the ground-state energy of a many-body Hamiltonian. Our methods give lower bounds on the ground-state energy, which are typically hard to come by, and might therefore be useful in a variety of related fields. 
%Since our approach connects some basic elements tightly together, such as graph parameters, joint numerical ranges, representations of (anti)commutation relations and ground-state energies, the impact of our results can directly influence the fields of graph theory, quantum information theory, (quasi-)Clifford algebra and quantum chemistry.
\end{abstract}

\date{16 April 2024} 

\maketitle

\section{Introduction}
\label{sec:intro}
Pauli strings, i.e.~tensor products of families of Pauli operators acting on multiple qubits, are undeniably among the most essential and omnipresent objects in the field of quantum information theory.  In addition to their role as unitary transformations, they commonly also serve as fundamental building blocks for constructing observables. Examples for such constructed observables are widespread, ranging from virtually every Hamiltonian that is considered in the field of quantum computing~\cite{kaye2006introduction,stolze2008quantum,briegel2009measurement,cafaro2011geometric,bharti2022noisy},  to typical measurements in quantum communication protocols~\cite{bennett2020quantum,hahn2019quantum,nema2020pauli,mastriani2021satellite}, to the rich theory of spin systems in condensed matter physics~\cite{luttinger1960ground,chapman2023unified,chapman2020characterization,kitaev2015simple,sachdev1993gapless,pirandola2020advances,fannes1992finitely}.

A pivotal object of investigation in this context is the set of jointly attainable expectation values. 
For a given  set of $n$ Pauli strings $\mathcal{S}=\{S_1,\dots,S_n\}$ acting on the Hilbert space $\mathcal{H}$,
we will consider
\begin{align}
    J(\mathcal{S})= \left\{ \left(\EXP{S_1} , \dots, \EXP{S_n} \right) | \rho\in\mathcal{D}(\mathcal{H}) \right\},\label{eq:jnrdef}
\end{align}
which is a compact convex subset of $\mathbb{R}^n$ (see Fig.~\ref{fig:paulis} for examples); here and subsequently, $\mathcal{D}(\mathcal{H})$ denotes the set of all states on the Hilbert space ${\cal H}$, i.e. density matrices, which are positive semidefinite and of unit trace.\newpage

For a general sequence {${\cal S}$} of observables, {$J(\mathcal{S})$ in Eq.} \eqref{eq:jnrdef} is commonly referred to as (convex) joint numerical range~\cite{halmos2012hilbert,gawron2010restricted,puchala2011product,dunkl2011numerical,gutkin2013joint,szymanski2018classification,szymanski2019geometric}, or convex support. For special sets of observables, this set may have however more specific names depending on the context. The most prominent example of this surely is the well-known Bloch ball/sphere, which can be understood as the joint numerical range of the three Pauli operators acting on a single qubit, i.e. $\mathcal{S}=\{X,Y,Z\}$.

\begin{figure}[ht]
    \centering
    \includegraphics[width=0.7\linewidth]{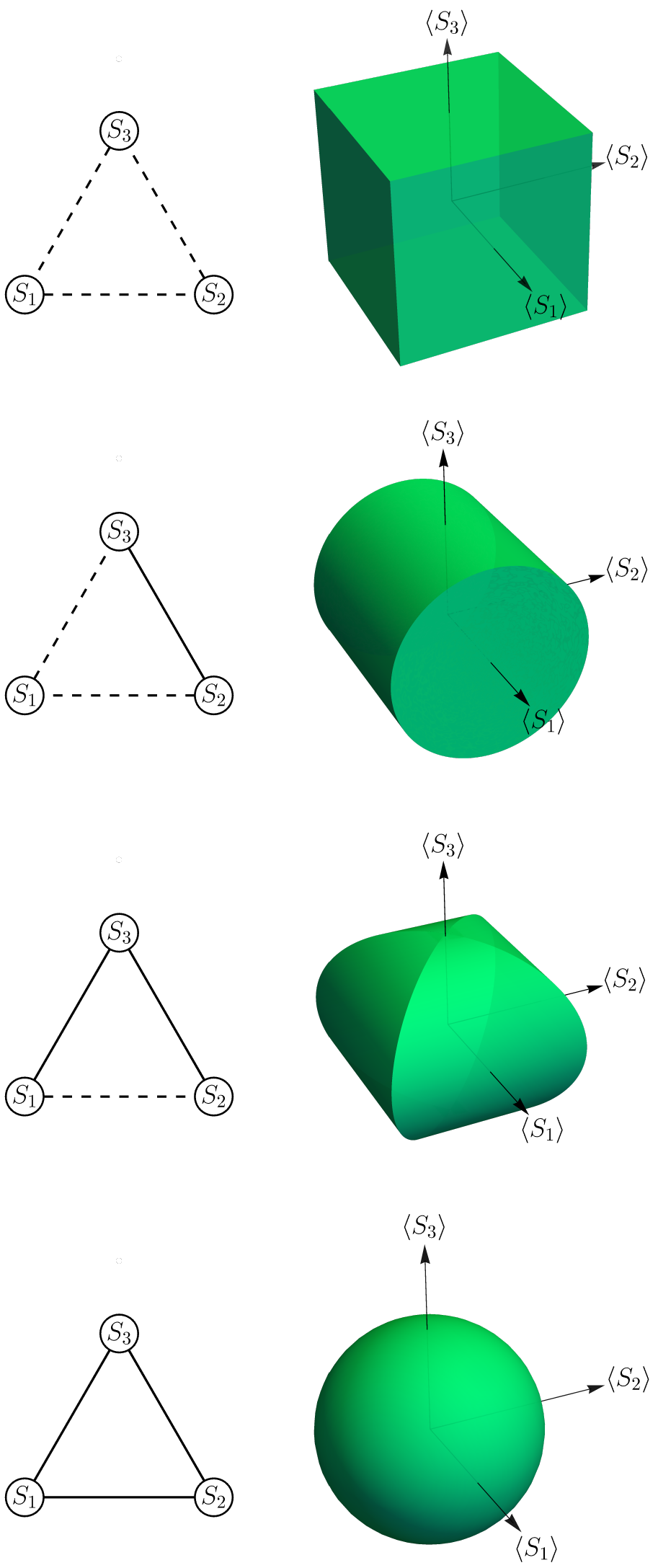}
    \caption{Joint numerical ranges of three Pauli strings whose commutation and anti-commutation relations are represented by the 4 possible graphs of three vertices. The real lines represent the anti-commutation relations and the dashed lines stand for commutativity relations. }
    \label{fig:paulis}
\end{figure}
For a more complex set of strings,
the geometry of $J(\mathcal{S})$ is {no} longer necessarily captured by a sphere. Nevertheless, one can {still} ask for the smallest radius of a sphere or more generally an ellipsoid encompassing it.  From a purely geometric perspective, this is the question  tackled in the present work.

The joint numerical range geometrically encodes a lot of relevant information about the interplay between observables and quantum states. Most strikingly, Hamiltonians that can be written as linear combinations of {elements in ${\cal S}$} can be directly understood as linear functionals acting on $J(\mathcal S)$. Correspondingly, the quest of finding ground-state energies and properties of states attaining them can be cast as the characterization of tangential hyperplanes on $J(\mathcal S)$. Moreover, upper and lower bounds to ground-state energies correspond to inner and outer approximations of $J(\mathcal S)$. 

This insight, however, directly reveals that the precise characterization of the joint numerical range can, even for simple objects like Pauli strings, be an intrinsically hard problem, i.e. at least as hard as solving ground state problems on multiple qubits. This clearly justifies {the} demand for practical approximation techniques.

For the ground-state energy problem itself, there is a large toolbox of {known} methods. Most of them provide lower bounds (inner bounds on $J(\mathcal{S})$) by restricting the optimization to a suitable variational class of states. Prominent examples to mention here are classical methods  building on matrix product states~\cite{perez2007matrix,schuch2008computational,fannes1992finitely}, as well as implementations on quantum computers where algorithms like the Quantum Approximate Optimization Algorithm (QAOA)~\cite{hadfield2019quantum,farhi2014quantum,binkowski2023elementary} or the Variational Quantum Eigensolver (VQE)~\cite{li2023quantum,kandala2017hardware} are very popular.

Methods for outer approximations, on the other hand, are rare. In contrast to variational methods, the handling of global structures is needed, which is typically much harder to achieve. 
Nevertheless, they can be of central importance. On one hand, outer bounds give complements for inner approximations and, by this, an interval for the overall accuracy of an approximation attempt. On the other, {the existence of} a bound that cannot be surpassed by any quantum state, also including those that are not in a variational class, can be critical in applications such as cryptography or entanglement detection. 

Notable examples for outer approximations can be found in  \cite{hastings2022optimizing,kull2022lower,eisert2023note}. Here, the ansatz is to exploit the algebraic structure of the problem for setting up hierarchies of non-commutative polynomial relaxations \cite{ligthart2023convergent,pironio2010convergent}. Even though the scope of these methods is to solve a particular ground state problem, the feasible set of such a relaxation, sometimes called the relaxation of the quantum set, can be seen as an outer approximation to the joint numerical range $J(\mathcal S)$. On the theoretical side, those methods usually come with a convergence guarantee. In practice, they however tend to show a very bad scaling behavior and typically saturate any practical limit of computational resources very quickly.   

In the present work, we move towards a {(semi-)}analytical understanding of $J(\mathcal S)$.
For an $m$-dimensional ellipsoid $\mathcal{E}_r(w)$ with principal axes of length $r/\sqrt{w_i}$, where $w=(w_1,\dots,w_m)$ is a positive weight vector, we ask for the smallest $r\geq 0$ such that
\begin{align}
    J(\mathcal{S}) \subseteq  \mathcal{E}_r(w)
\end{align}
holds.
{S}uch a minimal ellipsoid, or at least an upper bound on $r$, then allows to compute bounds on ground-state energies analytically.  This is especially interesting when {we deal} with high-dimensional objects, i.e. for long Pauli strings where the Hilbert space dimension scales exponentially.  We have to clarify that the deployment of a new numerical toolbox is, however, not the scope of the present work, and we leave this for a future investigation. Here, our main interest is to develop the foundations for a comprehensive understanding of the presented structures. 

Finding a valid $r$, the optimal value of which we will refer to as \emph{(generalized) radius}, corresponds to the non-linear optimization 
\begin{align}
    r^2\geq \sup_{\rho} \sum w_i \EXP{S_i}^2 \label{eq:minradius}
\end{align}
Formally, this can be seen as an instance of so-called spectrahedral inclusion problems~\cite{kellner2013containment,kellner2015semidefinite}, for which there are unfortunately no out-of-the-box solutions. 
Hence, a lot of the particular structure of the present problem has to be leveraged in the solution. 

The structure of minimal ellipsoids, as well as the structure of $J(\mathcal{S})$ itself (see Fig.~\ref{fig:paulis}), is closely related to a graph $G$ that encodes the commutation and anti-commutation relations within the strings of a set $\mathcal{S}$. 
This graph is also known under the names ``frustration graph''~\cite{chapman2020characterization,chapman2023unified}, ``anti-commutation graph''~\cite{gokhale2019minimizing} or ``anti-compatibility graph''~\cite{kirby2019contextuality} in the literature.
It can be shown~\cite{de2022uncertainty,hastings2022optimizing} (see Proposition~\ref{pro:main_results} below) that the minimal radius of an ellipsoid is lower bounded by the weighted independence number of $G$ and upper bounded by its weighted Lov\'{a}sz number. While being efficiently computable, the Lov\'{a}sz number {as an} upper bound is known to be generally not tight~\cite{hastings2022optimizing,de2022uncertainty}. The tightness of the lower bound was however outstanding. 
For the spheroidal case $w=(1,1,\dots)$ it was conjectured~\cite{de2022uncertainty}, and extensively  discussed during the coffee breaks of some workshops, that the independence number of $G$ describes the minimal radius. 
We show by an explicit example that this conjecture is false.

As a consequence, we introduce the minimal radius of an ellipsoid as a new graph parameter $\beta(G,w)$ and embark on its  exploration. We develop a number of equivalent formulations of it, we show that it is monotonic non-increasing when passing to an induced subgraph, multiplicative under the so-called lexicographic product, and several other properties. We calculate it for cyclic graphs $C_n$ and develop numerical tools to evaluate it for the complements of odd cycles $\overline{C}_{2n+1}$ supporting our conjecture that $\beta(\overline{C}_{2n+1})>2$ for all  $n\geq 3$. Lastly, we also elaborate on the connection of our results to the quest for  uncertainty relations between Pauli strings as their structure also depends directly on the structure of minimal ellipsoids \cite{de2022uncertainty,kaniewski2014entropic}. 

Our results connect multiple topics in quantum theory, graph theory, and algebra analysis, and so we believe that they might provide new insights in those fields as well.

\section{Operator representations of graphs, graph parameters and main results}
\label{sec:results}
In order to get a leverage on minimal ellipsoids we have to find global structures of the problem that can be efficiently used. Regarding a set $\mathcal{S}$ of Pauli strings as {the} representation of an algebra is along those lines. To this end, we introduce the frustration graph and consider an example. 

It is easy to see from the elementary properties of the four Pauli matrices (from now on including the identity), that each possible pair of Pauli strings either commutes or anti-commutes. 
For a set of strings, this information can be encoded in a graph $G$, where each vertex represents a string and an edge is drawn whenever two strings anti-commute. We will refer to this graph as frustration graph. 
Basic examples are the graphs shown in Fig.~\ref{fig:paulis}. For example, consider the third graph from above. This example was built by taking strings
\begin{align}
    S_1=X\otimes Y,\ 
    S_2= Z \otimes Z, \text{ and }
    S_3=Z\otimes Y.
    \label{eq:examplestrings}
\end{align}
The strings $S_1$ and $S_2$ commute, whereas the string $S_3$ anti-commutes with both of them. 
The joint numerical range is depicted next to it and arises from the intersection of two cylinders. For completeness, it makes sense to regard a cylinder as an asymptotic ellipsoid with one of its principal axes approaching infinity.  

Note that fixing a graph $G$ does not uniquely fix the set of strings representing it and there are more than only unitary degrees of freedom in them. In this sense, there are more sets of strings than there are graphs.
For example, any triple of strings will by construction give one of the graphs shown in Fig.~\ref{fig:paulis}. 

For the questions considered in this work, we can, however, make a basic proposition, formulated in more detail in Theorem~\ref{thm:standard}. It states that the radius of a minimal ellipsoid will be the same across all possible sets of strings representing the same graph.

\begin{proposition}
    Let $\mathcal{S}_1$ and $\mathcal{S}_2$ be sets of Pauli strings with the same frustration graph. Then we have that for any ellipsoid $\mathcal{E}_r(w)$,
    \begin{align}
        \mathcal{E}_r(w)\supseteq J(\mathcal{S}_1) \text{ if and only if } 
        \mathcal{E}_r(w)\supseteq J(\mathcal{S}_2).
    \end{align}
\end{proposition}
In other words, we can turn to investigating structures on a graph $G$ rather than on an explicit set of strings. 
This can be done by introducing an algebraic framework.
For a given graph $G$ we will consider operators with the abstract algebraic properties
\begin{itemize}
    \item[(i)] $S_i=S^*_i$,
    \item[(ii)] $S_i^2=\id$,
    \item[(iia)] $S_i^2\leq\id$,
    \item[(iii)] $\{S_i,S_j\}=0$ if $i\sim j$,
    \item[(iv)] $[S_i,S_j]=0$  $i \nsim j$,
\end{itemize}
where $i\sim j$ denotes vertices that are connected by an edge in a given graph $G$, i.e.~$i$ and $j$ are adjacent. 

It is clear that Pauli strings will obey all these properties. 
In order to investigate bounds on the radius of a minimal ellipsoid, we can, however, turn things around and consider all sets of operators fulfilling these conditions, or at least some of them. By this, we get a constructive tool for setting up relaxations of the joint numerical range. 
The detailed discussion of such operator representations of a graph will be the mathematical core of this work.
\begin{definition}
    For a given graph $G$, a set $\{S_i\}$ of operators in some Hilbert space $\mathcal{H}$ will be called 
    \begin{itemize}
        \item a self-adjoint unitary representation for anti-commutativity (SAURA) if (i), (ii), and (iii) hold. The set of all SAURAs will be denoted by $\saura(G)$.
        \item a self-adjoint unitary representation (SAUR) if (i), (ii), (iii){, and (iv)} hold.
        The set of all SAURs will be denoted by $\saur(G)$.
        \item a self-adjoint representation for anti-commutativity (SARA) if (i), (iia), and (iii) hold. 
        The set of all SARAs will be denoted by $\sara(G)$.
        \item a self-adjoint representation  (SAR) if (i), (iia),  (iii){, and (iv)} hold. 
        The set of all SARs will be denoted by $\sar(G)$.
    \end{itemize}
\end{definition}
It is clear from the definition above that the different sets of graph representations include each other, i.e. every SAUR is a SAR and a SAURA, and they are all a SARA. This can be captured in the diagram
\begin{center}
    \begin{tikzpicture}
      \matrix (m)
        [
            matrix of math nodes,
            row sep    = 0em,
            column sep = 0em
        ]
        {
            \sara(G)   &   \supseteq & \saura(G)\\
            \rotatebox[origin=c]{270}{$\supseteq$}   & &  \rotatebox[origin=c]{270}{$\supseteq$} \\
            \sar(G)  &    \supseteq &   \saur(G)\\
        };
   \end{tikzpicture}
\end{center}

One can also consider the {$C^*$}-algebras generated by the different sets of relations. Each of these algebras naturally comes with a state space, too, and hence also with a joint numerical range. These joint numerical ranges are by construction outer approximations to the joint numerical range of an explicit set of Pauli strings. 
The algebra generated by all SAURs is known as quasi-Clifford algebra \cite{gastineau1982quasi}.
It is a finite-dimensional algebra and can be seen as a generalization of the Clifford algebra in the sense that we obtain the Clifford algebras as the sub-cases with a fully connected graph. The representation theory of this object has been worked out in \cite{gastineau1981systems,gastineau1982quasi}, a brief summary can be found in \cite{hastings2022optimizing}. In a nutshell, we have that the Pauli strings with which we started are exactly generators of representations of this algebra. Bounds on minimal ellipsoids derived for this algebra are optimal. Bounds arising from considering the other algebras give relaxations, i.e. upper bounds on the radius. 

The algebra corresponding to all SAURAs can generally be infinite-dimensional, yet still separable. As we will show later (Lemma~\ref{lm:sar_dec}), the  algebra corresponding to SAR can be understood as tensor product of the quasi-Clifford algebra of the SAURs with a classical (i.e. commuting) algebra. An according result for the algebra corresponding to SAURA does however not hold. 

For a set $\Omega$ of operator representations of a graph, where each representation $\mathcal{S}$ acts on a Hilbert space $\mathcal{H}_\mathcal{S}$, we can extend the notion of a joint numerical range by considering
\begin{align}
    \JNR{\Omega}:= \left\{  \left(\mean{S_i}_\rho\right)\   |\ \mathcal{S}=\{S_i\}\in\Omega, \rho \in \mathcal{D}(\mathcal{H}_\mathcal{S}) \right\},
\end{align}
which is still a subset of $\mathbb{R}^n$, but now containing all possible tuples of expectations attainable by {the} set of  operators corresponding to a particular operator representation. 
In order to analyze ellipsoids, it makes sense to also consider the set of squared expectations
\begin{align}
        \SOS{\Omega}:= \left\{  \left(\mean{S_i}^2_\rho\right)\   |\ \mathcal{S}=\{S_i\}\in\Omega, \rho \in \mathcal{D}(\mathcal{H}_\mathcal{S}) \right\}.
\end{align}
Given this set, the minimal radius of an ellipsoid [recall \eqref{eq:minradius}] corresponding to some $w$ is the square root of the maximal value of a linear functional 
\begin{align}
    \sos{\Omega,w} := \sup_{(v_i) \in \SOS{\Omega}}\sum_i w_i v_i. \label{def:q}
\end{align}

In the spheroidal case, when all elements of $w$ are just $1$, we omit $w$ in the notation and write $\sos{\Omega}$. Furthermore, we denote by $\cSOS{\Omega}$ the convex hull of $\SOS{\Omega}$. If $\Omega$ contains only one element, say $\Omega=\{x\}$, we will abuse notation and write $q(x)$ instead of $q(\{x\})$. 

{For $\Omega$ a $\saur(G), \sar(G), \saura(G)$ or $\sara(G)$, it holds that
\begin{equation}
    \JNR{\Omega} = \{ (\pm \sqrt{q_i}) | (q_i) \in \SOS{\Omega}\},
\end{equation}
since $\{\pm S_i\} \in \Omega$ if $\{S_i\} \in \Omega$. In this case, the characterization of $\JNR{\Omega}$ is equivalent to the one of $\SOS{\Omega}$.
}
We remark that $\JNR{\Omega}$ is always convex and contains what is called the stable set polytope $\stab(G)$ that, in those four cases, is given by the convex hull of the set 
\begin{equation*}
  \{ (v_i) \mid v_i v_j = 0 \text{ if } i\sim j, v_i, v_j\in \{0,1\} \}.
\end{equation*}

Eq. \eqref{def:q} basically states that the full structure of ellipsoids encompassing the different $\JNR{\Omega}$, and by this also any $\JNR{\mathcal{S}}$ of interest, is encoded in the convex hulls of  the different $\SOS{\Omega}$. 

Now we state our main results and explain the structure of the paper.
\begin{proposition}{(Main results)}\label{pro:main_results}
Let $G$ be a graph. Let $\stab(G)$ be its stable set polytope and  let $\th(G)$ be its Lov\'{a}sz-theta body. For $\mathcal Q$ evaluated on different representations of $G$, as introduced above, we have the following ordering relations:
\begin{align}\label{eq:mainineq}
     \stab(G)\,
     \overset{\textcolor{lgreen}{(a)} }{\subseteq}\, \begin{matrix} \cSOS{ \saur(G) } \\ \rotatebox[origin=c]{90}{$=$}_{\textcolor{lgreen}{(b)} } \\  \cSOS{ \sar(G) } \end{matrix}\,
     \overset{\textcolor{lgreen}{(c)} }{\subseteq}\, \begin{matrix} \SOS{ \saura(G)  } \\ \rotatebox[origin=c]{90}{$=$}_{\textcolor{lgreen}{(d)} } \\  \SOS{ \sara(G)  } \end{matrix} 
     \overset{\textcolor{lgreen}{(e)} }{=}\,  \th(G)
\end{align}
with the Lov\'asz theta body defined as 
\begin{equation*}
    \th(G) = \{ |\langle \phi_0|v_i\rangle|^2 \mid \langle v_i|v_i\rangle = 1, \langle v_i|v_j\rangle = 0 \text{ if } i\sim j\},
\end{equation*}
where $|\phi_0\rangle$ is an arbitrary but fixed normalized vector in $\mathbf{R}^n$.
\end{proposition}

Consequently, for any non-negative weight vector $w$,
\begin{align}\label{eq:mainineq2}
     \alpha(G,w)\,
     {\le}\, \begin{matrix} \sos{ \saur(G), w} \\ \rotatebox[origin=c]{90}{$=$} \\  \sos{ \sar(G), w} \end{matrix}\,
     {\le}\, \begin{matrix} \sos{ \saura(G), w} \\ \rotatebox[origin=c]{90}{$=$} \\  \sos{ \sara(G), w} \end{matrix} 
     {=}\,  \vartheta(G,w),
\end{align}
where $\alpha(G,w)\! =\! q(\stab(G),w), \vartheta(G,w)\! =\! q(\th(G),w)$ are the weighted independence number and Lov\'asz number of $G$, and $\sos{ \saur(G), w}$ is denoted as $\beta(G,w)$ later.

The equality of $\SOS{ \saur(G) }$ and $\SOS{ \sar(G) }$ \textcolor{lgreen}{(b)}, as well as that of $\SOS{ \saura(G) }$ and $\SOS{ \sara(G) }$ \textcolor{lgreen}{(d)}, are outlined in Sec.~\ref{sec:sar} as Theorem~\ref{thm:sarsaur} and Theorem~\ref{thm:sarasaura}, respectively. Hence, we can focus the main part of our investigations to SAURA (Sec.~\ref{sec:saura}) and SAUR (Sec.~\ref{sec:saur}).

A relation between $\sos{ \saur(G), w}$ and $\vartheta(G,w)$ was already described in~\cite{de2022uncertainty,hastings2022optimizing}. From the algebraic perspective of this work, this relation is directly captured by the fact  that $\cSOS{ \saur(G) }$ is contained in $\th(G)$, their equivalence \textcolor{lgreen}{(b)} is summarized as Theorem~\ref{thm:thetavartheta} in Sec.~\ref{sec:saura}. In Sec.~\ref{beta}, we prove that the inclusion of $\stab(G)$ in $\SOS{ \saur(G) }$ can be strict for some graphs, which implies that $\beta(G,w)$ is a new graph parameter. For this, we have to firstly show in Theorem~\ref{thm:standard} in Sec.~\ref{sec:saur} that $\SOS{ \saur(G) } = \SOS{ \{S_i\} }$ for $\{S_i\}$ to be any SAUR of $G$.  We than continue to study the properties of the beta number under graph operations in Sec.~\ref{beta}. More explicitly, the properties for graph additions as in Theorem~\ref{thm:additionc} and Theorem~\ref{thm:additiond}, for graph products in Theorem~\ref{thm:lexi} and Theorem~\ref{thm:times}, for edge removal in Theorem~\ref{thm:eremoval}, and for cycle graphs in Theorem~\ref{thm:cycles}. At last, we apply those results to obtain better bounds than previous ones for uncertainty relations as in Thereom~\ref{thm:uncertainty}, and for ground-state energy as in Theorem~\ref{thm:ground}.

Establishing a bridge to the Lov\'asz-theta body $\th(G)$ and the Lov\'asz number, as described in Proposition \ref{pro:main_results}, is especially interesting from a numerical perspective since $\vartheta(G,w)$ can be computed via semi-definite programming (SDP). Using this as a tool for outer bounds is advantageous since the size of this SDP scales with $n$, {i.e.,} the number of vertices of $G$, and by this ultimately with the number of strings in a set $\mathcal{S}$. This has to be put into contrast with the problem of computing ground-state energies, i.e. hyperplanes of $J(\mathcal{S})$, whose size scales with the Hilbert space dimension $2^n$, where $n$ is the length of the strings.

Furthermore, for a given set of observables $\{S_i\}$, there is a complete hierarchy of SDP relaxations for $\sos{\{S_i\}}$ as explained in Appendix~\ref{sec:numerical}. Two practical see-saw methods for the numerical estimation are also supplemented in Appendix~\ref{sec:numerical}.  Finally, the estimation can be improved once we know the purity of the state as discussed in Appendix~\ref{sec:further}.

\section{Self-adjoint unitary representation for anti-commutativity}\label{sec:saura}
As introduced in Sec.~\ref{sec:results}, for a given graph $G$, a set $\{S_i\}$ of self-adjoint unitaries is said to be a self-adjoint unitary representation for anti-commutativity (SAURA) of $G$ if $i\sim j$ implies $\{S_i, S_j\} = 0$.
An essential task is to characterize $\sos{\mathcal{S}_a(G)}$ and $\SOS{\mathcal{S}_a(G)}$. 
\begin{lemma}~(Cf.~\cite{hastings2022optimizing})
For any graph $G$, it holds that
\begin{equation}\label{eq:sumofsquare2}
    \alpha(G) \le \sos{\mathcal{S}_a(G)} \le \vartheta(G).
\end{equation}
\end{lemma}
In fact, the relation $\sos{\{S_i\}} \ge \alpha(G)$ for any SAUR $\{S_i\}$ of a graph $G$  can be directly  proven by choosing the state $\rho$ as a common eigenstate of a set of the commuting observables from $S_i$, where $i$ ranges over an maximal independent set of $G$.
For any graph $G$, it holds that $\alpha(G) \le \vartheta(G)$. In the case that $G$ is a perfect graph, we have $\alpha(G) = \vartheta(G)$~\cite{knuth1994sandwich}. Hence, the upper bounds in Eq.~\eqref{eq:sumofsquare2} is tight for perfect graphs. Especially, in the case that $G$ is a clique graph with $n$ vertices where any two vertices are connected, the relation $\sos{\mathcal{S}_a(G)} = 1$ holds.  
For a given set of observables, numerical estimation of $\sos{\{S_i\}}$ can in principle provide a more exact bound. Different numerical methods to estimate $\sos{\{S_i\}}$ are presented in Appendix \ref{sec:numerical}.
Here we take the graph-theoretic approach, and provide an exact characterization of $\sos{\mathcal{S}_a(G)}$ and  $\SOS{\mathcal{S}_a(G)}$.

Graph-theoretic approach has been explored extensively in quantum contextuality~\cite{budroni2022kochen}, where the orthogonality representation (OR) of a graph is used~\cite{cabello2014graph}. For a given graph with $n$ vertices, a set of unit vectors $\{|v_i\rangle\}_{i=1}^n$ is said to be an OR of the graph $G$ if $\langle v_i|v_j\rangle = 0$ when $i\sim j$. 
An OR of a given graph implies an SAURA of the same graph.
\begin{lemma}\label{lm:vec2mat}
    For a given set of $d$-dimensional vectors $\{\langle v_i| = (v_{i1}, \ldots, v_{id})\}$, denote $S_i = \sum_k v_{ik} A_k$ where $\{A_k\}_{k=1}^d$ is a set of anti-commuting self-adjoint unitaries, it holds that $\{S_i, S_j^\dagger\}/2 = \langle v_i|v_j\rangle \id$. 
\end{lemma}
We remark that there are indeed SAURA of a graph that cannot be constructed from an OR. For example, the operators $\{X\id, Y\id, Z\id, ZZ\}$ and their anti-commutativity graph. Indeed, those four operators are linearly independent, but there is no operator which anti-commutes with all of them at the same time. Hence, there is no anti-commuting basis for those four operators. In this sense, there are more SAURAs than ORs of a given graph. 
Nevertheless, such a construction of SAURAs from ORs is the key step in the following proof.
\begin{theorem}\label{thm:thetavartheta}
    $\SOS{\saura(G)} = \th(G)$ for any graph $G$.
\end{theorem}
The proof is provided in Appendix~\ref{sec:proofs}, and the same for the other results following here.
Thus, we have a new physical explanation of the graph parameter $\vartheta(G)$.
\begin{figure}[htpb]
  \centering
  \includegraphics[width=0.18\textwidth]{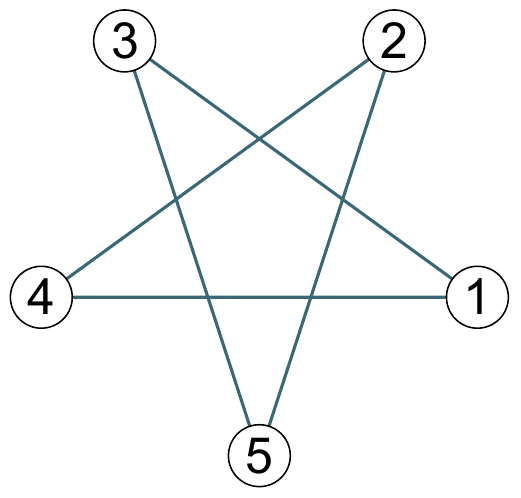}
\caption{Pentagon graph $C_5$.}
\label{fig:c5}
\end{figure}
\begin{example}\label{ex:pentagon}
    For the pentagon graph, its Lov\'{a}sz number can be achieved with the state $\langle u| = (1,0,0)$ and the following orthogonal representation $\{|v_i\rangle\}_{i=1}^5$, where
    \begin{align}
        \langle v_i| = \left(\tau, \tau'\cos(2\pi i/5), \tau'\sin(2\pi i/5) \right), 
    \end{align}
    and
    $\tau = (1/5)^{1/4}, \tau' = \sqrt{1-\tau^2}$.

    Hence, by choosing
    \begin{align} 
        &S_i = \sum_k v_{ik} \sigma_k, \forall i=1, \ldots, 5,\\
        &\rho = (\id + \sigma_1)/2,
    \end{align}
    where $\sigma_1 = X, \sigma_2 = Y, \sigma_3 = Z$ are Pauli matrices, then
    we have $\sum_i \mean{S_i}_\rho^2 = \sqrt{5}$. 
\end{example}

There is another concise proof of Theorem~\ref{thm:thetavartheta}. In fact, $\mean{S_i}_{\rho} = x\langle v_i|u\rangle$ if we take $\rho = (\id + x\sum_k u_k A_k)/d$, where $|x| \le 1$.  We note that $\rho \succeq 0$ is a legal state since the maximal eigenvalue of $\sum_k u_k A_k$ is no more than $1$. Thus, $\sum_iw_i\mean{S_i}_{\rho}^2 = x^2\sum_i w_i(\langle v_i|u\rangle)^2 = x^2\vartheta(G,w)$ by definition. By taking $x=\pm 1$, we complete the proof. 

Further results along the same line are provided in Appendix~\ref{sec:further}.

\section{Self-adjoint unitary representations}\label{sec:saur}
The set of operators, where any pair either commutes or anti-commutes, plays an important role as exemplified by the Pauli strings. The commutation and anti-commutation relations of such a set $\{S_i\}$ can be encoded into a so-called frustration graph $G$~\cite{chapman2020characterization,chapman2023unified}, where $i\sim j$ if $\{S_i, S_j\} = 0$, and $i\not\sim j$ if $[S_i, S_j]=0$.
By checking extensive examples, it is conjectured in Ref.~\cite{de2022uncertainty} that 
\begin{equation}\label{eq:oldConjecture}
    \sos{\{S_i\}} = \alpha(G).
\end{equation}
Whether Eq.~\eqref{eq:oldConjecture} can be violated is also an open question in Ref.~\cite{hastings2022optimizing}.
Conversely, for a given graph $G$, we can consider its representation by a set $\{S_i\}$ of self-adjoint unitaries, in the sense that $\{S_i,S_j\} = 0$ if $i\sim j$ and $[S_i,S_j] = 0$ if $i\not\sim j$. This representation is called self-adjoint unitary representation (SAUR)~\cite{samoilenko2012spectral}. By taking the graph-theoretic approach instead of starting from a special set, we denote $\beta(G,w) = \sos{\saur(G),w}$, where $\saur(G)$ is the set of all SAURs of $G$. The conjecture in Eq.~\eqref{eq:oldConjecture} is equivalent to $\beta(G) = \alpha(G)$. In Ref.~\cite{hastings2022optimizing}, no such an example is known that $\beta(G) > \alpha(G)$.
To continue, we first introduce the standard SAUR of a given graph, which is defined deductively. The standard SAUR can help us to reduce the complexity of considerations, since we only need to focus on the standard SAUR to obtain $\beta(G)$ as we will see later.
\begin{definition}
    For a given graph $G$ and one of its edges $(i_0,j_0)$, other vertices except $i_0$ and $j_0$ can be devided into four groups $V_0, V_1, V_2, V_3$, such that 
    \begin{itemize}
        \item $i\not\sim i_0$ and $i\not\sim j_0$ for any $i\in V_0$;
        \item $i\not\sim i_0$ and $i\sim j_0$ for any $i\in V_1$;
        \item $i\sim i_0$ and $i\sim j_0$ for any $i\in V_2$;
        \item $i\sim i_0$ and $i\not\sim j_0$ for any $i\in V_3$.
    \end{itemize}
    The subgraph $G'$ of $G$ with vertices in $\cup_{i=0}^4 V_i$ is said to be a Pauli-$(i_0,j_0)$-induced subgraph of $G$ if
    \begin{itemize}
        \item when $i\in V_{k_1}, j\in V_{k_2}$ where $k_1\neq k_2 \in \{1,2,3\}$, we have $i\sim j$ (or $i\not\sim j$) in $G'$ if and only if $i\not\sim j$ (or $i\sim j$) in $G$; 
        \item otherwise, $i\sim j$ in $G'$ if and only if $i\sim j$ in $G$.
    \end{itemize}
\end{definition}
\begin{definition}
    For a given graph $G$ and one of its edges $(i_0,j_0)$, denote $G'$ the Pauli-$(i_0,j_0)$-induced subgraph of $G$, from a standard SAUR $\{S'_i\}$ of $G'$, we call the following SAUR a standard one of $G$:
    \begin{equation}
      \left(\bigcup_{k=0}^3 \{\sigma_k\otimes S'_i\}_{i\in V_k}\right) \cup \{X_{i_0}\otimes \id, Z_{j_0}\otimes \id\}.
    \end{equation}
    If $G$ has no edge, we assign $\id$ or $1$ to all its vertices.
\end{definition}
If we take the pentagon $C_5$ in Fig.~\ref{fig:c5} as the original graph, and $(1,3)$ as the edge, then the Pauli-$(1,3)$-induced subgraph $G'$ is a triangle. Continually, the Pauli-$(4,5)$-induced subgraph $G''$ of $G'$ is just the vertex $2$. Hence, the  standard SAURs $\{S''_i\}, \{S'_i\}$ and $\{S_i\}$ of $G'', G'$ and $G$, respectively, are 
\begin{align}
    &S''_2 = 1, S'_2 = Y, S'_4 = X, S'_5 = Z,\\
    &S_1 = X\id, S_2 = \id Y, S_3 = Z\id, S_4 = Z X, S_5 = X Z, 
\end{align}
where we have omitted the symbol of tensor product, and $X\id$ means $X\otimes \id$ etc.
\begin{theorem}[Cf.~\cite{samoilenko2012spectral}]\label{thm:saur}
    For a given graph $G$, one SAUR $\{S_i\}$ and one standard SAUR $\{\bar{S}_i\}$ of $G$, there is a unitary $U$ such that $US_iU^\dagger = \bar{S}_i \otimes D_i$, where $\{D_i\}$ is a set of commuting self-adjoint unitaries.
\end{theorem}
Though different sequences of edges lead to different standard SAUR, Theorem~\ref{thm:saur} implies that they are in the same dimension and connected by unitaries.
The standard SAUR is succinct, however, it loses the information of the symmetry in the graph. To reflect the structure of the graph, we introduce the edge SAUR.
\begin{definition}
    For a given directed graph $\hat{G}$ with $n$ vertices and the edge set $E$, the set of self-adjoint operators $\{A_i\}_{i=1}^n$ with $A_i = \otimes_{e\in E} O_{e,i}$ is called the edge SAUR of $\hat{G}$ , where $O_{e,i} = X$ if $i$ is the start of $e$, and $O_{e,i} = Z$ if $i$ is the end of $e$, otherwise, $O_{e,i} = \id$.
\end{definition}
For an undirected $G$, we can lift the concept of an edge SAUR by simply choosing directions. The resulting representation is indeed unique. Switching between different choices of directions results into permuting $X$'s and $Z$'s. Since every edge corresponds to a single qubit in a edge SAUR, this operation corresponds to a unitary transformation. 
For different SAURs of the same graph $G$, their joint expectation values are related.
\begin{theorem}\label{thm:standard}
    For a given graph $G$, $\sos{\{S_i\}} = \sos{\{\bar{S}_i\}}$ where $\{S_i\}$ is a SAUR of $G$ and $\{\bar{S}_i\}$ is a standard one.
\end{theorem}
Hence, we have $\beta(G) = \sos{\{\bar{S}_i\}}$ where $\{\bar{S}_i\}$ is any standard SAUR of $G$.
A similar result of Theorem~\ref{thm:standard} for the weighted version can be proven in the same way. Consequently, $\SOS{\saur(G)} = \SOS{\{{S}_i\}}$ where $\{{S}_i\}$ is any SAUR of $G$. However, $\JNR{\saur(G)}$ might be strictly larger than $\JNR{\{S_i\}}$ due to the sign of each expectation value. To recover the whole set of $\JNR{\saur(G)}$, it is enough to consider the complete SAUR as defined below.
\begin{definition}
    For a given graph $G$ with $n$ vertices and its standard SAUR $\{\bar{S}_i\}$, the set $\{S_i\}$ consisting of
    \begin{equation}
        S_i = \bar{S}_i \otimes \Big(\bigotimes_{k=1}^n Z^{\delta_{ik}} \Big)
    \end{equation}
    is said to be a complete SAUR, where $\delta_{ik}=1$ if $i=k$, otherwise, $\delta_{ik}=0$.
\end{definition}
As we can see, the auxiliary part $\{\otimes_{k=1}^n Z^{\delta_{ik}}\}$ can recover all signs in $\{-1, 1\}^{\otimes n}$. Effectively, this provides a cover of $\JNR{\saur(G)}$ with multiple copies of $\JNR{\{\bar{S}_i\}}$. Generally, $\JNR{\{\bar{S}_i\}}$ can be neither point or reflection symmetric, but $\JNR{\saur(G)}$ has both point and reflection symmetries.

%%%%%%%%%%%%%%%%%%%%%%%%%%%%%%%%%%%%%%%%%%%%%%%%%%%%%%%
\section{The beta parameter}\label{beta}
For the characterization of the graph parameter $\beta(G)$, we consider some properties of it in this section. First, we show that $\beta(G)$ is indeed different from $\alpha(G)$, although it can occur that $\alpha(G)=\beta(G)$ for certain graphs.

\begin{corollary}\label{cor:gh10}
  For the graph $G_{10}$ and $C_5$ in Fig.~\ref{fig:gh10}, 
  \begin{equation}
      \beta(G_{10}) = \alpha(G_{10}) = \beta(C_5) = \alpha(C_5) = 2.
  \end{equation}
\end{corollary}
\begin{figure}[htpb]
  \centering
  \includegraphics[width=0.275\textwidth]{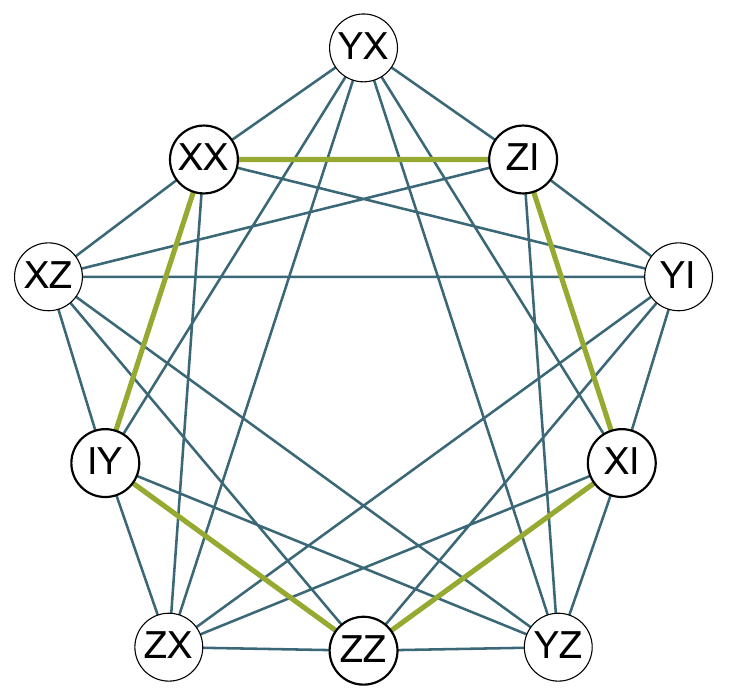}
  \caption{A graph $G_{10}$ with $10$ vertices and one of its standard realizations, which contains a pentagon $C_5$ as an induced subgraph (green, thick).}
  \label{fig:gh10}
\end{figure}
\begin{figure}[ht]
    \centering
    \includegraphics[width=0.55\linewidth]{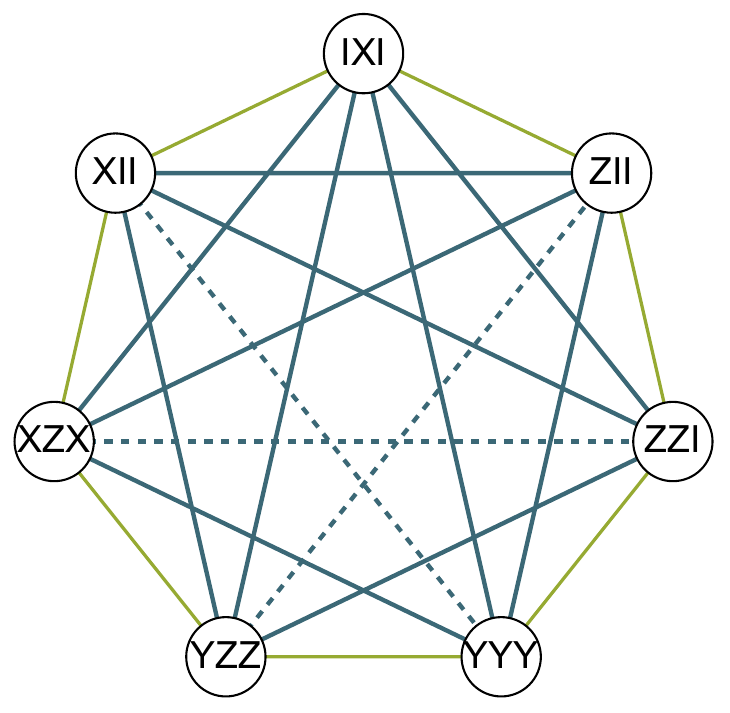}
    \caption{Graph of the counterexample and its SAUR. A pair of observables anti-commute when the corresponding vertices are connected by an edge in the graph of an anti-heptagon $\bar{C}_7$ (all blue lines) and commute when connected by an edge in the graph of a heptagon $C_7$ (green, thin). The subgraph of $\bar{C}_7$ (blue, not including the dotted lines ) is named as $G_7$.}
    \label{fig:antihepta}
\end{figure}

However, the conjecture in Eq.~\eqref{eq:oldConjecture} is not true generally. A simple counterexample is the anti-heptagon and its standard SAUR as shown in Fig.~\ref{fig:antihepta}. To be more explicit, the seven operators are
\begin{align}
  S_1 &= ZZ\id, 
  &S_2 &= Z\id\id, 
  &S_3 &= \id X\id,
  &S_4 &= X\id\id, \nonumber\\ 
  S_5 &= XZX,
  &S_6 &= YZZ,
  &S_7 &= YYY.&
\end{align}
It can be checked by hand that these operators form one standard SAUR of $\bar{C}_7$, and $\alpha(\bar{C}_7)=2$.

Let $\rho=|v\rangle\langle v|$ be the state that corresponds to the largest eigenvector of $\sum_i S_i$. With a little bit more hand work, or by using a computer algebra system, one can check that
\begin{align}
    \sum \langle S_i\rangle_\rho^2 = (9+4\sqrt{2} ) /7 \approx 2.09384 > 2=\alpha(\bar{C}_7),
\end{align}
which disproves the conjecture in Eq.~\eqref{eq:oldConjecture}. 
Besides, $\vartheta(\bar{C}_7) = 1 + 1/\cos(\pi/7) \approx 2.10992$.
Thus, in general, $\alpha(G) \le \beta(G) \le \vartheta(G)$ are indeed three different graph parameters.

Then we consider properties of the beta number under graph operations, e.g., the addition, graph products such as lexicographic product and XOR product. Those properties are helpful for the estimation of the beta number for large graphs, which might be impossible with the numerical methods.
\paragraph{Additions and products of graphs}
\begin{theorem}\label{thm:additionc}
  For a given graph $G$ which can be divided into two subgraphs $G_1, G_2$ where all vertices in $G_1$ are connected with all vertices in $G_2$, then $\beta(G) = \max \{\beta(G_1), \beta(G_2)\}$.
\end{theorem}
\begin{corollary}
  If we add one new vertex to a graph $G$ and result in a graph $G'$ in the way that the new vertex is connected to all vertices in $G$, then $\beta(G') = \beta(G)$.
\end{corollary}
\begin{theorem}\label{thm:additiond}
  For a given graph $G$ which can be divided into two subgraphs $G_1, G_2$ where any vertex in $G_1$ is disconnected with any vertex in $G_2$, $\beta(G) = \beta(G_1) + \beta(G_2)$.
\end{theorem}
For two given graphs $G_1$ and $G_2$, we denote by $G_1[G_2]$ their lexicographic product, whose vertex set is the Cartesian product of the graphs' vertex sets and then $(i_1,j_1) \sim (i_2, j_2)$ if $i_1\sim i_2$, or $j_1\sim j_2$ when $i_1 = i_2$.
\begin{theorem}\label{thm:lexi}
  For two given graphs $G_1$ and $G_2$, $\beta(G)$ is multiplicative under the lexicographic product: $\beta(G_1[G_2]) = \beta(G_1) \beta(G_2)$.
\end{theorem}

For any large graph with decomposition into small graphs with known beta numbers through the two addition operations and lexicographic product, its exact beta number can be obtained. For example, if we take the lexicographic product of five $\bar{C}_7$'s, that is, $G = \bar{C}_7^{[~]5}$, then $\beta(G) = \beta(\bar{C}_7)^5 \approx 40.2452$. However, $\alpha(G) = \alpha(\bar{C}_7)^5 = 32$ and $\vartheta(G) = \vartheta(\bar{C}_7)^5 \approx 41.8144$ since the latter two parameters are multiplicative under the lexicographic product, too~\cite{geller1975chromatic,lovasz1979shannon,cubitt2014bounds}.
Hence, the integer parts of $\beta(G)$, $\alpha(G)$ and $\vartheta(G)$ can be all different. This answers an open question in Ref.~\cite{hastings2022optimizing} in the negative: there are indeed graphs with beta number strictly larger than the independence number, and the gap between them can be arbitrarily large.

The tensor product of systems is often relevant in quantum mechanics. Denote $G$ the anti-commutativity and commutativity graph corresponding to the tensor product of the SAURs of $G_1$ and $G_2$. We can directly verify that $G$ is the XOR product of $G_1, G_2$, that is, $(i_1,j_1) \sim (i_2, j_2)$ if and only if only one of $i_1\sim i_2$ and $j_1\sim j_2$ holds. In this case, denote $G = G_1 \times G_2$.
\begin{theorem}\label{thm:times}
For any pair of graphs $G_1$ and $G_2$,\linebreak
$\beta(G_1\times G_2) \ge \beta(G_1) \beta(G_2)$.
\end{theorem}

\paragraph{Removal of edges}
The removal of one edge is also one basic graph operation, which can relate different graph products.
One important property shared by the independence number and the Lov\'{a}sz number is that they do not decrease under the edge-removal. However, this does not hold for the beta number. Here we take $\bar{C}_7$ and its subgraph $G_7$ (see Fig.~\ref{fig:antihepta}) as an 
example, where the aforementioned properties of the beta number play a role. 
\begin{theorem}\label{thm:eremoval}
    $\beta(G_7) = 2 < \beta(\bar{C}_7) $.
\end{theorem}
\begin{proof}
    Notice that $G_7$ is isomorphic to an induced subgraph of $C_5[K_2]$, where $K_2$ is just one edge. Thus, 
    \begin{equation}
        \beta(G_7) \le \beta(C_5)\beta(K_2) = 2.
    \end{equation}
    On the other hand, $\beta(G_7) \ge \alpha(G_7) = 2$, which completes the proof.
\end{proof}
Although the beta number is between the independence number and the Lov\'{a}sz number, its behavior under edge-removal is rather strange.
Nevertheless, the beta number does not increase under the vertex-removal, same as the independence number and the Lov\'{a}sz number. More explicitly, $\beta(G') \le \beta(G)$ if $G'$ is an induced subgraph of $G$.

\paragraph{Cycles and anticycles}
For a perfect graph $G$, we know that $\alpha(G) = \vartheta(G) = \alpha^\star(G)$ which implies that $\alpha(G) = \beta(G) = \vartheta(G) = \alpha^\star(G)$. For imperfect graphs, odd cycles and odd anticycles are basic building blocks. 
As it turns out, the following holds.
\begin{theorem}\label{thm:cycles}
$\beta(C_m) = \alpha(C_m)$, which consequently implies that  $\cSOS{\saur(C_m)} = \stab(C_m)$.
\end{theorem}

Numerically, we have $\vartheta(\bar{C}_{2n+1}) > \beta(\bar{C}_{2n+1}) > \alpha(\bar{C}_{2n+1}) = 2$ for $n\le 10$, see Fig.~\ref{fig:antic} for more details; for any graph $G$ with no more than $9$ vertices, if $\beta(G) > \alpha(G)$, then $G$ has either $\bar{C}_7$ or $\bar{C}_9$ as an induced subgraph.  
\begin{figure}
    \centering
    \includegraphics[width=0.45\textwidth]{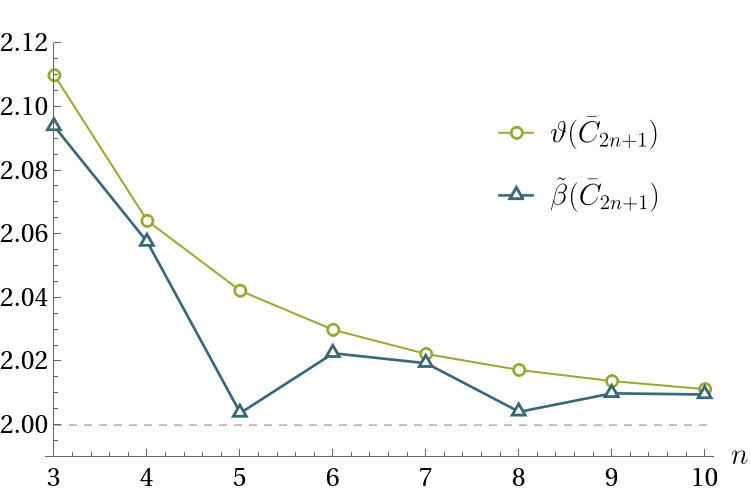}
    \caption{The estimation $\tilde{\beta}(\bar{C}_{2n+1})$ of $\beta(\bar{C}_{2n+1})$ in comparison with $\vartheta(\bar{C}_{2n+1})$. For each estimation, we have used the second see-saw method with $500$ rounds of iteration.}
    \label{fig:antic}
\end{figure}
Though these observations might suggest the conjecture that \emph{$\cSOS{\saur(G)} \supsetneq \stab(G) $ if and only if $G$ has some $\bar{C}_{2n+1}$ as an induced subgraph where $n \ge 3$}, it is refuted by a very recent counterexample~\cite{hastings2024limitations}.

%%%%%%%%%%%%%%%%%%%%%%%%%%%%%%%%%%%%%%%%%%%%%%%%%%%%%%%%%%%%%%%%
\section{self-adjoint representations and return to joint numerical range beyond $\mathbf{\pm 1}$-observables}
\label{sec:sar}
Whether in SAURA or in SAUR, we have only considered self-adjoint unitaries, which can limit the range of applications. We generalize our setting to non-unitary operators in this section.
\begin{definition}
  For a given graph $G$, a set of operators $\{A_i\}$ is said to be a self-adjoint representation for anti-commutativity (SARA) of $G$ if each $A_i$ is self-adjoint, $A_i^2\leq \id$, $\{A_i, A_j\} = 0$ when $i\sim j$. Furthermore, if $[A_i, A_j] = 0$ whenever $i\not\sim j$, $\{A_i\}$ is said to be a self-adjoint representation (SAR) of $G$.
\end{definition}
For a given graph $G$, denote $\sara(G)$ the set of all its SARAs, and $\sar(G)$ the set of all its SARs. By definition, $\saura(G) \subseteq \sara(G)$, and $\saur(G) \subseteq \sar(G)$. Surprisingly, maximizing over $\sara(G)$ instead of $\saura(G)$ does not result in a larger value, same for $\sar(G)$ and $\saur(G)$.
To be more explicit:
\begin{theorem}\label{thm:sarasaura}
  $\SOS{\sara(G)} = \SOS{\saura{G}} = \th(G)$.
\end{theorem}
\begin{theorem}\label{thm:sarsaur}
  $\SOS{\sar(G)} = \SOS{\saur{G}}$.
\end{theorem}

A set is said to be star-convex if all the points, which are on the line segment between the origin and any point in the set, are in the set, too. As it turns out, $\SOS{\saur(G)}$ is indeed star-convex. More details are provided in Appendix~\ref{sec:proofs}.

\section{Applications}
As an application, we provide bounds for sum uncertainty relations among sets of observables with certain anti-commutation or commutation relations. 

\begin{theorem}\label{thm:uncertainty} 
For a given set of observables $\{A_i\}_{i=1}^n$ that is a SARA of $G$, we have
\begin{equation}\label{eq:lovaszu}
  \sum_i \Delta^2(A_i) \ge \lambda_{\min} - \vartheta(G,\lambda),
\end{equation}
where $\lambda_{\min}$ is the minimal singular value of $\sum_i A_i^2$, and $\lambda = (a_1^2, \ldots, a_n^2)$ with $a_i$ the maximal eigenvalue of $A_i$. Besides, if  $\{A_i\}$ is a SAR of $G$, then
\begin{equation}\label{eq:betau}
  \sum_i \Delta^2(A_i) \ge \lambda_{\min} - \beta(G,\lambda).
\end{equation}
\end{theorem}
For a given set of observables $\{A_i\}_{i=1}^n$, the estimation of $\lambda_{\min}$ is a relatively easy problem. In the case that each $A_i$ has only two outcomes $\pm a_i$, then $\lambda_{\min} = \sum_i a_i^2$. 
In comparison with the similar application on uncertainty relations in Ref.~\cite{de2022uncertainty}, our results are not limited to dichotomic observables. Besides, the inequality \eqref{eq:betau} is tighter in general.

Our results can also be used to estimate the ground-state energy, which is of high interest in quantum many-body systems~\cite{luttinger1960ground,li2023quantum}. 
For a given frustration graph $G$ with $n$ vertices, the dimension of the system to realize it is exponential of $n$~\cite{gastineau1982quasi}. This can make the problem  notoriously challenging to solve. 
\begin{theorem}\label{thm:ground}
    For a given set of Pauli strings $\{A_i\}_{i=1}^n$ whose frustration graph is $G$, 
\begin{align}
  \Big(\sum_i a_i \mean{A_i}_\rho\Big)^2 &\le \min_w  \Big( \sum_i a_i^2/w_i\Big)\beta(G,w)\label{eq:minsar}\\
  &\le \min_w  \Big( \sum_i a_i^2/w_i\Big)\vartheta(G,w),\label{eq:minsara}
\end{align}
for any state $\rho$ and any real coefficients $a_i$'s.
\end{theorem}
\begin{proof}
    By employing the Cauchy-Schwarz inequality, we have
    \begin{align}
        \Big(\sum_i a_i \mean{A_i}\Big)^2 &\le \Big( \sum_i a_i^2/w_i\Big) \Big( \sum_i w_i \mean{A_i}^2\Big) \label{eq:foropt}\\
        &\le \Big( \sum_i a_i^2/w_i\Big) \beta(G,w),\label{eq:beta_bound}
    \end{align}
    where $w_i$'s are positive. Thus, the inequality still holds when we take the minimization over possible $w$.
\end{proof}
Especially, we can set $w_i = |a_i|^t$ for $t=0, 1, 2$.
In the case that $t=0$, we recover the result in Ref.~\cite{hastings2022optimizing} with Eq.~\eqref{eq:minsara} for Pauli strings, which is less tight than the one in Eq.~\eqref{eq:minsar} in general. However, when some $a_i$'s are much larger than others, much better performance can be obtained by taking $t=1$ or $t=2$. For example, $a_1=1$ and all the others are $0$, then the results for $t=0$ and $t=1$ are $\vartheta(G)$ and $1$ respectively. And $\vartheta(G)$ can be much larger than $1$ for a big graph $G$. In general, a tighter bound can be obtained from Eq.~\eqref{eq:foropt} than the case where $w_i = |a_i|^t$ for $t=0, 1, 2$. More precisely, the bound reads
\begin{equation}
    \min_w \max_{v\in Q}\Big( \sum_i a_i^2/w_i\Big) \Big( \sum_i w_i v_i\Big),
\end{equation}
where $Q=\th(G)$ or $Q=\cSOS{\saur(G)}$. In the case that $G$ is either a perfect graph or a cycle, $Q=\cSOS{\saur(G)} = \stab(G)$. This becomes a Min-Max problem of size $n$, which is relatively much easier than the original problem. 

We take the Hamiltonian  $H = \sum_{i=1}^7 A_i$ on $14$ qubits as the first example, where
\begin{align}   
A_1 \! &=\! Z_{1}Z_{2}Z_{3}Z_{4}, &A_2 \! &=\! Z_{5}Z_{6}Z_{7}Z_{8},\nonumber\\ 
A_3 \! &=\! X_{1}Z_{9}Z_{10}Z_{11}, &A_4 \! &=\! X_{2}X_{5}Z_{12}Z_{13},\nonumber\\
A_5 \! &=\! X_{3}X_{6}X_{9}Z_{14}, &A_6 \! &=\! X_{4}X_{7}X_{10}X_{12},\nonumber\\
A_7 \! &=\! X_{8}X_{11}X_{13}X_{14}.
\end{align}
This set $\{A_i\}_{i=1}^7$ is in fact the edge SAUR of $\bar{C}_7$ as in Fig.~\ref{fig:antihepta}.
Direct calculation shows that minimal and maximal eigenvalues of $H$ are $\mp 3.828427$, respectively. This agrees with the upper bound in Eq.~\eqref{eq:beta_bound} by setting the $w_i$'s equal to $1$, i.e., $\sqrt{7\beta(\bar{C}_7)} = 1+2\sqrt{2}$.

We take the Hamiltonian  $H = \sum_{i=1}^9 A_i$ on $27$ qubits as the second example, where
\begin{align}
A_1\! &=\! Z_{1}Z_{2}Z_{3}Z_{4}Z_{5}Z_{6}, 
& A_2 \! &=\! Z_{7}Z_{8}Z_{9}Z_{10}Z_{11}Z_{12},\nonumber\\
A_3 &\! =\! X_{1}Z_{13}Z_{14}Z_{15}Z_{16}Z_{17}, & A_4 \! &=\! X_{2}X_{7}Z_{18}Z_{19}Z_{20}Z_{21},\nonumber\\
A_5 \! &=\! X_{3}X_{8}X_{13}Z_{22}Z_{23}Z_{24}, & A_6 \! &=\! X_{4}X_{9}X_{14}X_{18}Z_{25}Z_{26},\nonumber\\
A_7 \! &=\! X_{5}X_{10}X_{15}X_{19}X_{22}Z_{27}, & A_8 \! &=\! X_{6}X_{11}X_{16}X_{20}X_{23}X_{25},\nonumber\\
A_9 \! &=\! X_{12}X_{17}X_{21}X_{24}X_{26}X_{27}.
\end{align}
This set $\{A_i\}_{i=1}^9$ is in fact the edge SAUR of $\bar{C}_9$.
Direct calculation of the minimal and maximal eigenvalues of $H$ is hard since its dimension is $2^{27}$.
An easy numerical calculation shows that $\beta(\bar{C}_9) = 2.057505$, which implies that the maximal singular value of $H$ is bounded by $\sqrt{9\beta(\bar{C}_9)} = 4.303201$, where we have utilized Eq.~\eqref{eq:beta_bound} by setting $w_i$'s equal to $1$. This bound is in fact tight, as we can verify by converting $\{A_i\}$ into the standard SAUR with an auxiliary system as in Theorem~\ref{thm:saur}.

For a general frustration graph with $n$ vertices, the dimension of the edge SAUR is $2^{|E|}$, where $|E|$ is the number of edges in the frustration graph, and typically $|E| = O(n^2)$. Hence, the direct evaluation of the ground-state energy is in dimension $2^{O(n^2)}$. The estimation as in Eq.~\eqref{eq:beta_bound} deals only with a graph with $n$ vertices, which could be much simpler.

\section{Conclusion and discussion}
\label{sec:conclusion}
The deep connection between sums of squares of expectations of Pauli strings on the one hand, and parameters of the corresponding frustration graph on the other, has already been observed and successfully used in ~\cite{de2022uncertainty,hastings2022optimizing}. In \cite{de2022uncertainty} it was used to derive upper and lower bounds on uncertainty relations in terms of the independence number and in terms of the Lov\'{a}sz number, respectively. Conceptually similar results where employed by \cite{hastings2022optimizing} in the context of many-body Hamiltonians.  
Those initial findings opened up a broader perspective with lots of new questions to answer, setting one of the starting points for the present work. 

We investigate the joint numerical range of a set of Pauli strings, as the general object that encodes the answers to a whole cornucopia of interesting questions, and of which we have here only scratched the surface.  
In our setting, the previously discussed bounds on sums of squares can be cast in the nice geometrical picture of ellipsoids encompassing $J(S)$. A central finding here, and another starting point for our study, is the insight that the structure of these ellipsoids is solely determined by the frustration graph and not by a particular realization of Pauli strings itself.  Due to this we introduce and investigate the invariant $\beta(G)$ as a new graph parameter.  
This shift of attention towards the graph then brought us directly to employing  an algebraic description that takes a particular graph as input and provides an axiomatically defined algebra as output. These algebras are known as quasi-Clifford algebras, and their investigation goes back to works of Gastineau-Hills  \cite{gastineau1981systems,gastineau1982quasi} in the 1980s, which gave a solid foundation of results on which we could build.
Furthermore, dropping parts of these axioms then naturally leads to relaxations and hence ultimately to outer bounds on the joint numerical range. This approach gives rise to what we call here SARA, SAURA, and SAR. 
In conclusion, we find that taking a graph-theoretic approach instead of focusing on a set of exact operators opened up an unexpectedly rich new perspective leading in many directions.

First, the comparison with the graph-theoretic approach in other fields builds up a channel, through which we can convert known results into the field of joint numerical range. In current work, the graph-theoretic approach in contextuality leads to the proof of the tight bound for the SAURA case, and the generalization of the SAURA case. 
Secondly, the graph-theoretic approach can pick up the most relevant representations and discover new crucial gradients like purity. In the SAUR case, only after proving that the upper bound is always achievable by the standard representation, it becomes possible to enumerate graphs to compare the upper bound with other graph parameters. As it turned out, the upper bound beta number in the SAUR case is indeed a new graph parameter, which is different from both the independence number and the Lov\'{a}sz number. According to the known evidence, the beta number can be used as a quite good approximation of the independence number. Thirdly, the graph-theoretic approach provides another level, i.e., the graph level, to study the joint numerical range. By considering graph operations and special graphs, we can characterize the beta number and get the beta number for large graphs, or large sets of operators, which might be impossible for direct numerical calculations. Hence, the graph-theoretic approach not only deepens and expands the field of joint numerical range, it also connects to other fields like quantum contextuality and graph theory. In our current work, we have also developed numerical methods for the estimation of the upper bound from below and from above. 
At last, we generalized this approach to general self-adjoint operators which do not need to be unitaries.

However, there are still many open problems remaining, some of which we highlight here:
\begin{itemize}
    \item How to calculate the beta number, especially approximate it from above, more efficiently?
    \item The beta number is closely related to the algebra generated by the standard SAUR. How to characterize the beta number from the algebraic graph theory? In the case that the graph has the whole algebra as its SAUR, Ref.~\cite{de2018state} already has a complete characterization. 
    \item Denote $\otimes$ the OR product of graphs. Does $\beta(G_1\otimes G_2) = \beta(G_1) \beta(G_2)$ hold? If this is true, then the beta number could be closely related to the Shannon capacity of the graph~\cite{fritz2021unified}.
    \item Could the graph-theoretic approach be applied to the variance-based criteria for quantum correlations, like the one for entanglement~\cite{hofmann2003violation}?
    \item How could we develop the graph-theoretic approach for other kinds of uncertainty relations? The fractional packing number $\alpha^\star(G)$ might play an important role, which is defined as the maximum of $\sum_{i=1}^n x_i$ where $x_i\ge 0$ and $\sum_{i\in C} x_i \le 1$ for any clique $C$ in graph $G$. For example, for the quantum entropic uncertainty relation of a set of anti-commuting self-adjoint unitaries, it is known~\cite{wehner2008higher} that $\sum_{i\in C} [1-H(S_i|\rho)] \le 1$ holds and $0 \le H(S_i|\rho) \le 1$, where $H(S|\rho)$ is the Shannon entropy of the statistics from the two-outcome measurement $S$ with the state $\rho$. This leads to $\sum_{i=1}^n H(S_i|\rho) \ge n-\alpha^\star(G)$, where $G$ is the anti-commutativity graph of $\{S_i\}$.
\end{itemize}

{\it Acknowledgments.---}
The authors thank Carlos de Gois, Felix Huber, Gereon Ko{\ss}mann, Kiara Hansenne, Nikolai Wyderka,  Otfried G\"{u}hne, Reinhard F. Werner and Ties Ohst for insightful discussions, and in particular Dolly Vook von Katterbar and Pete Club for invaluable comments on ground level optimization.
Z.P.X. acknowledges support from {National Natural Science Foundation of China} (grant no.~12305007), 
Anhui Provincial Natural Science Foundation (grant no.~2308085QA29), the Deutsche Forschungsgemeinschaft (DFG, German Research Foundation, project numbers 447948357 and 440958198), the Sino-German Center for Research Promotion (project M-0294), the ERC (Consolidator Grant 683107/TempoQ), and the Alexander von Humboldt Foundation.
R.S. acknowledges financial support by the Quantum Valley Lower Saxony, Quantum Frontiers, and by the BMBF projects ATIQ, QuBRA and CBQD.
A.W. was supported by the European Commission QuantERA grant ExTRaQT (Spanish MICIN project PCI2022-132965), by the Spanish MINECO (projects PID2019-107609GB-I00 and PID2022-141283NB-I00) with the support of FEDER funds, by the Spanish MICIN with funding from European Union NextGenerationEU (PRTR-C17.I1) and the Generalitat de Catalunya, by the Spanish MTDFP through the QUANTUM ENIA project: Quantum Spain, funded by the European Union NextGenerationEU within the framework of the ``Digital Spain 2026 Agenda'', by the Alexander von Humboldt Foundation, and by the Institute for Advanced Study of the Technical University Munich.

\appendix
\onecolumngrid

\section{Further results concerning SAURA}
\label{sec:further}
\subsection{Tighter bound based on purity} Note that the purity can be written as $\tr \rho^2  = (1+x^2)/d$, which provides the intuition that the purity of the state could affect the joint expectation values. A direct observation is that $\sum_i \mean{S_i}^2 = 0$ always holds for the maximally mixed state. By making use of the information of purity, we can improve the estimation of joint expectation values.
\begin{theorem}
    For a given set of $d$-dimensional observables $\{S_i\} \in \saura(G)$, and a state $\rho$, we have
    \begin{equation}\label{eq:mixedState}
        \sum_i \mean{S_i}_{\rho}^2 \le \min\{[d(\tr\rho^2) - 1]\vartheta(G), \vartheta(G)\}. 
    \end{equation}
    Furthermore, for any given graph $G$ and the purity of state, the upper bound in Eq.~\eqref{eq:mixedState} is always tight.
\end{theorem}

The discussion at the beginning of this section works as a constructive proof.
We remark that $[d(\tr\rho^2) - 1]\vartheta(G)\le \vartheta(G)$ always hold when the dimensional is $2$, like the case in Example~\ref{ex:pentagon}. 
Notice that $\tr \rho^2$ is related to the linear entropy of the state, the inequality shows how the entropy of the state affects the joint expectation values. It is interesting to see that the linear entropy of the state affects the joint expectation values only when the linear entropy is large enough, i.e., when $\tr \rho^2 \le 2/d$. This happens when the temperature of the thermal state is high, or the system is highly entangled with the environment.

\subsection{Relaxation of anti-commutation relation}
The anti-commutation relation of operators leads to the orthogonality of them in the sense of trace product. That is, if $\{S_i, S_j\} = 0$, then $\tr S_iS_j = 0$. However, the converse is not necessarily true, e.g., $S_1 = XX$ and $S_2 = ZZ$. For convenience, denote $|S\rangle$ as the vector obtained by flattening the operator $S$ row by row. With this notation, $\tr S_iS_j = \langle S_i|S_j\rangle$.  As we will see later, $\{S_i, S_j\} = 0$ also implies $\tr S_i = \tr S_j = 0$. Hence, $\mean{S_i}_{\rho} = \langle S_i|\tilde{\rho}\rangle$, where $\tilde{\rho} = \rho - \id/d$ and $d$ is the dimension. Notice that $\langle \tilde{\rho}|\tilde{\rho}\rangle = \tr\rho^2 - 1/d$, $\langle S_i|S_i\rangle = d$. By comparing with the definition of Lov\'{a}sz number, we have a generalization of Theorem~\ref{thm:thetavartheta}.
\begin{theorem}\label{thm:orthogonalLovasz}
    For a given graph $G$ and its $d$-dimensional orthogonality representation with $\{S_i\}$ such that $\langle S_i|S_i\rangle = d$ and $\langle S_i|S_j\rangle = 0$ if $i\sim j$, then
    \begin{equation}
        \sum_i \mean{S_i}_{\rho}^2 \le [d(\tr\rho^2) - 1]\vartheta(G).
    \end{equation}
\end{theorem}
We have two remarks: first, the constructive proof of Theorem~\ref{thm:thetavartheta} implies that the bound in Theorem~\ref{thm:orthogonalLovasz} is tight whenever $\tr \rho^2 \le 2/d$.
Secondly,
Lemma~\ref{lm:vec2mat} and the vectorization of matrices, Theorem~\ref{thm:thetavartheta} and Theorem~\ref{thm:orthogonalLovasz} give hints to the similarity of the role of $\vartheta(G)$ in quantum contextuality and joint expectation values. 

\section{Numerical methods}\label{sec:numerical}
The independence number $\alpha(G)$ is an important graph parameter, which has application in the characterization of channel capacity. The calculation of $\alpha(G)$ is NP-hard~\cite{karp1972reducibility}, the calculation of $\vartheta(G)$ is just a semi-definite programming. Thus, $\vartheta(G)$ can be used as an approximation of $\alpha(G)$. 
%Shannon capacity, the difficulty. Quantum algorithm for the estimation in the future.
Since $\beta(G)$ is a tighter upper bound of $\alpha(G)$ than $\vartheta(G)$, efficient methods to estimate $\beta(G)$ are necessary. 
As we have proven, $\beta(G) = \sos{\{\bar{S}_i\}}$ where $\{\bar{S}_i\}$ is any standard SAUR of $G$, a more general problem is to estimate $\sos{\{S_i\}}$ for a given set of $\{S_i\}$. For example, there might be only some anti-commutation relations in $\{S_i\}$. In this section, we provide two efficient see-saw methods to give lower bounds of $\beta(G)$, and one complete hierarchy of semi-definite programming to approximate $\beta(G)$ from upper bound.

\subsection{Lower bounds}
We notice that 
\begin{align}\label{eq:withSymm}
    \sos{\{S_i\}} &= \max_{\rho}\sum_{i=1}^n \tr\left[(\rho\otimes\rho) (S_i\otimes S_i) \right]\\ 
    &= \max_{\rho}\tr\left[(\rho\otimes\rho) \sum_{i=1}^n (S_i\otimes S_i) \right]. 
\end{align}
On the one hand, we have
\begin{align}
   \tr\left[(\rho_1\otimes\rho_2) \sum_{i=1}^n (S_i\otimes S_i) \right] &\le \left( \sum_{i=1}^n \langle S_i\rangle_{\rho_1}^2 \sum_{i=1}^n \langle S_i\rangle_{\rho_2}^2 \right)^{1/2}\nonumber\\ 
   &\le \sos{\{S_i\}}.
\end{align}
On the other hand,
\begin{align}
   \sos{\{S_i\}} \le \max_{\rho_1,\rho_2} \tr\left[(\rho_1\otimes\rho_2) \sum_{i=1}^n (S_i\otimes S_i) \right].
\end{align}
Consequently, we have
\begin{align}\label{eq:product_state}
   \sos{\{S_i\}} = \max_{\rho_1,\rho_2} \tr\left[(\rho_1\otimes\rho_2) \sum_{i=1}^n (S_i\otimes S_i) \right].
\end{align}
For a given $\rho_1$, the optimal $\rho_2$ corresponds to the eigenstate with the maximal singular value of $\sum_{i=1}^n (\tr\rho_1S_i) S_i$. Similarly, for a given $\rho_2$, the optimal $\rho_1$ corresponds to the eigenstate with the maximal singular value of $\sum_{i=1}^n (\tr\rho_2S_i) S_i$. Hence, we can use a See-Saw method to estimate $\sos{\{S_i\}}$, where each step is only an SVD decomposition. We remark that this See-Saw method can be generalized for any polynomials of mean values by considering $\rho_1, \ldots, \rho_k$, where $k$ is the order of the polynomial.

Another observation is that
\begin{align}
    \sos{\{S_i\}} = \max_{\rho,\ c\ s.t.\ ||c||_2=1} \left(\sum_i c_i \mean{S_i}_\rho \right)^2,
\end{align}
which leads to another See-Saw method for the lower bound. As it turns out, for a given $\rho$, the optimal vector $c$ is the normalized vector of $(\mean{S_i}_\rho)_{i=1}^n$. For a given vector $c$, the optimal $\rho$ corresponds to the eigenstate with the maximal singular value of $\sum_{i=1}^n c_i S_i$.

Since it is unnecessary to require that $S_i^2 = \id$ in those two methods, they can be naturally extended to the weighted version, i.e., $q(\{S_i\},w)$.

\subsection{Upper bounds}
According to Eq.~\eqref{eq:product_state} and the linearity on $\rho_1\otimes \rho_2$ as a whole state, we have
\begin{align}
   \sos{\{S_i\}} = \max_{\gamma \in \operatorname{SEP}} \tr\left[\gamma \sum_{i=1}^n (S_i\otimes S_i) \right],
\end{align}
where $\operatorname{SEP}$ is the set of separable states.

Our first observation is that the maximum can always be achieved in the case that $\gamma$ is a pure state, and the set of pure separable states can be fully characterized by the PPT condition and the rank-$1$ constraint.

Thus, we can reformulate the optimization into a rank-constrained problem,
\begin{align}
    \sos{\{S_i\}} = & \max_{\gamma} \tr\left[ \gamma \sum_i (S_i\otimes S_i) \right]\nonumber\\
    \text{such that } & \tr \gamma = 1, \gamma \succeq 0, \gamma^{T_2} \succeq 0, \label{eq:sdpCond}\\
                      & F_{12}\gamma = \gamma,\label{eq:symmetryCond}\\
                      & \operatorname{rank}\gamma = 1,\label{eq:rankCond}
\end{align}
where $F_{12}$ is the swap operator $\sum_{ij} |ij\rangle\langle ji|$, $T_2$ means the partial transpose on the second party. 

According to Eq.~\eqref{eq:withSymm}, the condition in Eq.~\eqref{eq:symmetryCond} can be added without changing the outcome. 

As proposed in Ref.~\cite{yu2022quantum}, there is a complete hierarchy of relaxation with semi-definite programming for the rank-constrained problem, which leads to such a complete hierarchy for the problem in our consideration.
However, this technique is not so practical here. Denote $d$ the dimension of $S_i$, then the dimension of $\gamma$ is $d^2$. The size of the matrix on the $k$-th level is then $d^{2k}$.
Even if $d=8$ like in the standard SAUR of $\bar{C}_7$, by taking $k=2$, we have $d^{2k} = 4096$, which is quite hard for a normal computer. For practical purposes, we propose the following relaxation of the rank-$1$ constraints:
\begin{align}
    \sos{\{S_i\}} = & \max_{\gamma} \tr\left[ \gamma \sum_i (S_i\otimes S_i) \right]\nonumber\\
    \text{such that } & \tr \gamma = 1, \gamma \succeq 0, \gamma^{T_2} \succeq 0, \nonumber\\
                      & F_{12}\tau = F_{23}\tau = \tau,\nonumber\\
                      & \tr_3 \tau = \gamma,
\end{align}
which can be seen as the $3/2$-level of the hierarchy. This technique is special for our case since the state $\gamma$ is already two copies of the state in the system of $\{S_i\}$. 

Another approach to make relaxation of $\operatorname{SEP}$ is to add more semi-definite conditions like the PPT condition and linear conditions like entanglement witness. We recommend the reader to look at Appendix B of Ref.~\cite{wyderka2020characterizing} for detailed discussions. The conditions in Eq.~\eqref{eq:extraConds} in the main text are such an example.

\section{Proofs of main results}
\label{sec:proofs}
%%%%%%%%%%%%%%%%%%%%%%%%%%%%%%%%%%%%%%%%%%%%%%%%%%%%%%%%%%%%%%%%%%%%%%%%%%%%%%%%%%%%%%%
\begin{theorem}\label{thm:thetavartheta2}
    $\SOS{\saura(G)} = \th(G)$ for any graph $G$.
\end{theorem}
\begin{proof}
  It is equivalent to prove that $\sos{\mathcal{S}_a(G),w} = \vartheta(G,w)$ for any non-negative weight vector $w$. Eq.~\eqref{eq:sumofsquare2} includes already the result that $\sos{\mathcal{S}_a(G),w} \le \vartheta(G,w)$ when all the elements of $w$ are just $1$. For the general non-negative weight vector $w$, we prove it later in Sec.~\ref{sec:sar}.

  To show that $\sos{\mathcal{S}_a(G),w} \ge \vartheta(G,w)$, we construct an exact $\{S_i\} \in \saura(G)$ such that $\sos{\{S_i\},w} = \vartheta(G,w)$.
  
  Denote $\{|v_i\rangle\}_{i=1}^n$ and $|u\rangle$ the OR of the graph $G$ and the state such that $\vartheta(G,w) = \sum_{i=1}^n w_i |\langle v_i|u\rangle|^2$, which can be assumed to be real without loss of generality.
  Denote $\{A_i\}_{i=1}^r$ a set of $d$-dimensional normalized traceless observables satisfying $\{A_i, A_j\}/2 = \delta_{ij} \id$, where $r$ is the dimension of $\{|v_i\rangle\}_{i=1}^n$. 
  By setting $S_i = \sum_{k=1}^r v_{i,k} A_k$, the $S_i$'s are Hermitian and $\{S_i, S_j\}/2 = \langle v_i|v_j\rangle \id$, which implies that $\{S_i\}$ is an SAURA of $G$. 
  
  For a given state $\rho$, denote $\mathcal{M}_\rho$ the matrix whose $(i,j)$-th element is $\sqrt{w_iw_j}\mean{\{S_i, S_j\}/2}_\rho$. In this special case, $\mean{\{S_i, S_j\}/2}_\rho = \langle v_i|v_j\rangle$ which is independent of the exact state $\rho$. Then $\lambda_{\max}(\mathcal{M}_\rho) = \vartheta(G,w)$. Denote $|a\rangle$ the eigenvector of $\mathcal{M}_\rho$ corresponding to the maximal eigenvalue, we have $\langle a|\mathcal{M}_\rho|a\rangle = \vartheta(G,w)$. 

  Denote by $|s\rangle$ the eigenstate of $\sum_i a_i \sqrt{w_i}S_i$ corresponding to the maximal eigenvalue and $\sigma = |s\rangle\langle s|$. Then
  \begin{align}\label{eq:lowerboundLovasz}
     \sum_{i} w_i\mean{S_i}_{\sigma}^2 &\ge \left(\sum_{i} a_i \sqrt{w_i}\mean{S_i}_{\sigma} \right)^2\nonumber\\
     & = \left(\mean{\sum_{i} a_i \sqrt{w_i}S_i}_{\sigma} \right)^2\nonumber\\
     & = \mean{\left(\sum_{i} a_i \sqrt{w_i}S_i \right)^2}_{\sigma}\nonumber\\
     & = \vartheta(G,w),
  \end{align}
  where the first line is from Cauchy-Schwarz inequality since $|a\rangle$ is normalized, the third line is by the definition of $\sigma$ and the last line is from the definition of $|a\rangle$ and the fact that $\mathcal{M}_\rho$ is independent of the state $\rho$. Finally, we have $\sos{\mathcal{S}_a(G),w} \ge \vartheta(G,w)$ and complete the proof.
\end{proof}
We remark that the construction in Lemma~\ref{lm:vec2mat} is crucial for the proof, as the last line in Eq.~\eqref{eq:lowerboundLovasz} may not hold for a general SAURA.

\begin{theorem}
    For a given graph $G$, $\sos{\{S_i\}} = \sos{\{\bar{S}_i\}}$ where $\{S_i\}$ is a SAUR of $G$ and $\{\bar{S}_i\}$ is a standard one.
\end{theorem}
\begin{proof}
  From the convexity of $\sum_i \mean{S_i}^2$, we know that we only need to prove for the case that $\rho$ is a pure state $|\psi\rangle\langle \psi|$.
 Since $D_i$ commutes with each other, we can assume $D_i$'s are diagonal matrices. Denote $d_2$ the dimension of $D_i$'s, then we have the decomposition
 \begin{equation}
   U|\psi\rangle = \sum_{i=1}^{d_2} \sqrt{p_i} |\phi_i\rangle \otimes |i\rangle,
 \end{equation}
 where $p_i \ge 0$ and $\sum_i p_i = 1$.

 Hence,
 \begin{align}
   \mean{S_i} &
   = \sum_{kl} \sqrt{p_ip_j} \langle \phi_k|\bar{S}_i|\phi_l\rangle \langle k|D_i|l\rangle\nonumber\\
   &= \sum_k p_k \langle \phi_k|s_{ik} \bar{S}_i|\phi_k\rangle,
 \end{align}
 where $s_{ik} \in \{-1, 1\}$ is the $k$-th diagonal element in $D_i$.

 Then we have
 \begin{equation}
   \mean{S_i}^2 = \left( \sum_k p_k \langle \phi_k|s_{ik} \bar{S}_i|\phi_k\rangle \right)^2 \le \sum_k p_k \langle \phi_k|\bar{S}_i|\phi_k\rangle^2,
 \end{equation}
 which implies that
 \begin{align}
   \sum_i \mean{S_i}^2 &\le \sum_k p_k \left( \sum_i |\langle \phi_k|\bar{S}_i|\phi_k\rangle|^2 \right)\nonumber\\ 
   &\le \max_k \sum_i \langle \phi_k|\bar{S}_i|\phi_k\rangle^2.
 \end{align}
 Thus, $\sos{\{S_i\}} \le \sos{\{\bar{S}_i\}}$.

 On the other hand, denote $|\phi\rangle$ the optimal state for $\sos{\{\bar{S}_i\}}$ and $|\psi_0\rangle$ the common eigenstate for $D_i$'s, then 
 \begin{align}
   \sum_i \mean{S_i}^2_\rho = \sum_i \langle \phi|\bar{S}_i|\phi\rangle^2,
 \end{align}
 where $\rho = |\psi\rangle\langle \psi|$ and $ |\psi\rangle = U^\dagger[|\phi\rangle \otimes |\psi_0\rangle]$. Thus, we have $\sos{\{S_i\}} \ge \sos{\{\bar{S}_i\}}$. This finishes the proof.
\end{proof}

\begin{corollary}\label{cor:gh102}
  For the graph $G_{10}$ and $C_5$ in Fig.~\ref{fig:gh10}, 
  \begin{equation}
      \beta(G_{10}) = \alpha(G_{10}) = \beta(C_5) = \alpha(C_5) = 2.
  \end{equation}
\end{corollary}
\begin{proof}
  Notice that $\beta(G) \ge \alpha(G)$, $\beta(G_{10}) \ge \beta(C_5)$ and $\alpha(G_{10}) = \alpha(C_5) = 2$, we only need to prove $\beta(G_{10}) = 2$.
  According to Theorem~\ref{thm:standard}, it is sufficient to consider the standard SAUR of $G_{10}$ as the one in Fig.~\ref{fig:gh10} and denote it by $\{S_i\}_{i=1}^{10}$.
  
  Because of the convexity of $\sum_{i} \mean{S_i}_{\rho}^2$ in terms of state $\rho$, we only need to consider the state to be a pure $4$-dimensional state
  \begin{align}
    |\psi\rangle = \{&\cos\theta_1, e^{it_1}\sin\theta_1\cos\theta_2, e^{it_2}\sin\theta_1\sin\theta_2\cos\theta_3,\nonumber\\ 
    &e^{it_3}\sin\theta_1\sin\theta_2\sin\theta_3\}.
  \end{align}
  A direct calculation shows that
  \begin{equation}\label{eq:gh10}
    \sum_{i=1}^{10} \mean{S_i}^2 = 2.
  \end{equation}
\end{proof}
The fact that $\beta(C_5) = 2$ has also been proven in Ref.~\cite{hastings2022optimizing} in another approach.
We have two remarks: first, for any $4$-dimensional pure state, Eq.~\eqref{eq:gh10} is equivalent to
\begin{equation}
  \mean{\id X}^2 + \mean{\id Z}^2 + \mean{XY}^2 + \mean{YY}^2 + \mean{ZY}^2 = 1.
\end{equation}
By permuting $X, Y, Z$ and the parties, we can obtain other equalities.
Secondly, the standard SAUR of $G_{10}$ cannot be generated by the construction in Lemma~\ref{lm:vec2mat}, since there is no operator anti-commuting with $X\id, Y\id, Z\id$ at the same time, meanwhile, the dimension of the linear span of all the operators in this standard SAUR is $10$.

\begin{theorem}
  For a given graph $G$ which can be divided into two subgraphs $G_1, G_2$ where all vertices in $G_1$ are connected with all vertices in $G_2$, then $\beta(G) = \max \{\beta(G_1), \beta(G_2)\}$.
\end{theorem}
\begin{proof}
  For a given SAUR of $G$, we label the operators for $G_1$ as $\{A_i\}$, the ones for $G_2$ as $\{B_j\}$.

  On the one hand, for any state $\rho$ (for convenience, we omit the state $\rho$ in the mean value), we have
  \begin{align}
    &\sum_i \mean{A_i}^2 + \sum_j \mean{B_j}^2\nonumber\\
    = &\max_{x,y} \left[\mean{\sum_i x_i A_i}^2 + \mean{\sum_j y_j B_j}^2\right]\nonumber\\
    = &\max_{x,y,t} \left[\mean{t_1 \sum_i x_i A_i + t_2\sum_j y_j B_j}^2\right]\nonumber\\
    \le &\max_{x,y,t} \mean{\left(t_1 \sum_i x_i A_i + t_2\sum_j y_j B_j\right)^2},
  \end{align}
  where $x, y, t$ are unit real vectors.

  By definition of $\{A_i\}$ and $\{B_j\}$, we have
  \begin{align}
    &\max_{x,y,t} \mean{\left(t_1 \sum_i x_i A_i + t_2\sum_j y_j B_j\right)^2}\nonumber\\
    = &\max_{x,y,t} \left[t_1^2\mean{\left(\sum_i x_i A_i\right)^2} + t_2^2 \mean{\left(\sum_j y_j B_j\right)^2}\right]\nonumber\\
    = &\max_{x,y,t} \max\left\{\mean{\left(\sum_i x_i A_i\right)^2}, \mean{\left(\sum_j y_j B_j\right)^2}\right\}\nonumber\\
    \le &\max_{x,y,t} \max\{\beta(G_1), \beta(G_2)\}\nonumber\\
    = &\max\{\beta(G_1), \beta(G_2)\},
  \end{align}

  By definition,
  \begin{align}
    \beta(G) &= \max_{\rho, \{A_i\},\{B_j\}} \left[ \sum_i \mean{A_i}^2_\rho + \sum_j \mean{B_j}^2_\rho \right],
  \end{align}
  which implies that $\beta(G) \le \max\{\beta(G_1), \beta(G_2)\}$.
  
  On the other, $\beta(G) \ge \max\{\beta(G_1), \beta(G_2)\}$, which completes the proof.
\end{proof}
\begin{corollary}
  If we add one new vertex to a graph $G$ and result in a graph $G'$ in the way that the new vertex is connected to all vertices in $G$, then $\beta(G') = \beta(G)$.
\end{corollary}
\begin{theorem}
  For a given graph $G$ which can be divided into two subgraphs $G_1, G_2$ where any vertex in $G_1$ is disconnected with any vertex in $G_2$, $\beta(G) = \beta(G_1) + \beta(G_2)$.
\end{theorem}
\begin{proof}
  Without loss of generality, we assume that the state $\rho$ and the SAUR $\{A_i\}$ result in $\beta(G_1)$, and the state $\sigma$ and the SAUR $\{B_j\}$ result in $\beta(G_2)$. Then $\{A_i \otimes \id_B\} \cup \{\id_A \otimes B_j\}$ is a SAUR of $G$. Direct calculation shows that
  \begin{align}
    \sum_i\! \mean{A_i\!\otimes\! \id}_{\rho\otimes \sigma}^2\! +\! \sum_j \mean{\id\otimes B_j}_{\rho\otimes \sigma}^2\! =\! \beta(G_1)\! +\! \beta(G_2).
  \end{align}
  Besides, for any SAUR $\{\tilde{A}_i\} \cup \{\tilde{B}_j\}$ of $G$ and any state $\tau$, we have
  \begin{align}
    &\sum_i \mean{\tilde{A}_i}_{\tau}^2 + \sum_j \mean{\tilde{B}_j}_{\tau}^2\nonumber\\ 
    \le &\max_{\rho, \{A_i\}} \sum_i \mean{A_i}_{\rho}^2 + \max_{\sigma,\{B_j\}} \sum_j \mean{B_j}_{\sigma}^2\nonumber\\
    = &\beta(G_1) + \beta(G_2).
  \end{align}
  In total, we have $\beta(G) = \beta(G_1) + \beta(G_2)$ by definition.
\end{proof}
For two given graphs $G_1$ and $G_2$, we denote by $G_1[G_2]$ their lexicographic product, whose vertex set is the Cartesian product of the graphs' vertex sets and then $(i_1,j_1) \sim (i_2, j_2)$ if $i_1\sim i_2$, or $j_1\sim j_2$ when $i_1 = i_2$.
\begin{theorem}
  For two given graphs $G_1$ and $G_2$, $\beta(G)$ is multiplicative under the lexicographic product: $\beta(G_1[G_2]) = \beta(G_1) \beta(G_2)$.
\end{theorem}
\begin{proof}
  Denote $\{A_{ij}\}$ any SAUR of $G=G_1[G_2]$, where $A_{ij}$ represents the vertex in $G$ which corresponds to $i$ in $G_1$ and $j$ in $G_2$. Denote $\bar{A}_i = \sum_j x_{ij}A_{ij}/\lambda_i$, where $\{x_{ij}\}_{j}$ is the normalized vector of $\{\mean{A_{ij}}\}_{j}$, and $\lambda_i$ is the maximal eigenvalue of $\sum_j x_{ij}A_{ij}$. We remark that $\{\bar{A}_i\} \in \sar(G_1)$, and $\{A_{ij}\}_j \in \saur(G_2)$ for all $i$.

  By definition,
  \begin{align}\label{eq:product_upper_bound}
    \beta(G) &= \max_{\rho, \{A_{ij}\}\in \saur(G)} \sum_{i} \sum_j \mean{A_{ij}}_{\rho}^2\nonumber\\
%             &\le \sum_i \max_{\rho, \{A_{ij}\} \in \mathcal{S}(G)} \sum_j \mean{A_{ij}}_{\rho}^2\\
             &= \max_{\rho, \{A_{ij}\} \in \saur(G) }\sum_i \mean{\sum_j x_{ij} A_{ij}}_{\rho}^2\nonumber\\
             &= \max_{\rho, \{A_{ij}\} \in \saur(G) }\sum_i \lambda_i^2 \mean{\bar{A}_i}_{\rho}^2\nonumber\\
             &\le \beta(G_2) \max_{\rho, \{A_i\} \in \sar(G_1)} \sum_i \mean{A_i}_{\rho}^2\nonumber\\
             &= \beta(G_2) \beta(G_1),
  \end{align}
  where the inequality in the fourth line is from the following fact: 
  \begin{equation}
    \lambda_i^2 = \max_{\rho}\mean{\sum_j x_{ij} A_{ij}}_{\rho}^2 \le \beta(G_2),
  \end{equation}
  and the last equality is proven in Sec.~\ref{sec:sar}.

  On the other hand, without loss of generality, we assume that the state $\rho$ and the SAUR $\{A_i\}$ result in $\beta(G_1)$, the state $\sigma$ and the SAUR $\{B_j\}$ result in $\beta(G_2)$.
  Denote
  \begin{equation}
    A_{ij} = A_i \otimes \left[\bigotimes_{k\in G_1} B_j^{\delta_{ik}}\right],
  \end{equation}
  where $\delta_{ik} = 1$ if $k = i$, otherwise, $\delta_{ik}=0$.
  By construction, $\{A_{ij}\} \in \mathcal{S}(G)$. Let $\tau := \rho \otimes \sigma^{\otimes n_1}$, where $n_1$ is the number of vertices in $G_1$. Then we have
  \begin{equation}
    \sum_{ij} \mean{A_{ij}}_{\tau}^2 = \sum_{ij} \mean{A_i}_{\rho}^2 \mean{B_j}_{\sigma}^2 = \beta(G_1)\beta(G_2),
  \end{equation}
  concluding the proof.
\end{proof}

For any large graph with decomposition into small graphs with known beta numbers through the two addition operations and lexicographic product, its exact beta number can be obtained. For example, if we take the lexicographic product of five $\bar{C}_7$, that is, $G = \bar{C}_7^{[\ ]5}$, then $\beta(G) = \beta(\bar{C}_7)^5 \approx 40.2452$. However, $\alpha(G) = \alpha(\bar{C}_7)^5 = 32$ and $\vartheta(G) = \vartheta(\bar{C}_7)^5 \approx 41.8144$ since those two parameters are also multiplicative under the lexicographic product~\cite{geller1975chromatic,lovasz1979shannon,cubitt2014bounds}. Hence, the integer part of $\beta(G)$, $\alpha(G)$ and $\vartheta(G)$ can be all different. This closes the open question in Ref.~\cite{hastings2022optimizing} with the answer: there are indeed graphs with beta number strictly larger than the independence number, the gap between them can even be large.

The removal of one edge is also one basic graph operation, which can relate different graph products.
One important property shared by the independence number and the Lov\'{a}sz number is that they do not decrease under the edge-removal. However, this does not hold for the beta number. Here we take $\bar{C}_7$ and its subgraph $G_7$ (see Fig.~\ref{fig:antihepta}) as an example. 
\begin{theorem}
    $\beta(G_7) = 2 < \beta(\bar{C}_7) $.
\end{theorem}
\begin{proof}
    Notice that $G_7$ is isomorphic to an induced subgraph of $C_5[K_2]$, where $K_2$ is just one edge. Thus, 
    \begin{equation}
        \beta(G_7) \le \beta(C_5)\beta(K_2) = 2.
    \end{equation}
    On the other hand, $\beta(G_7) \ge \alpha(G_7) = 2$, which completes the proof.
\end{proof}
Although the beta number is between the independence number and the Lov\'{a}sz number, its behavior under edge-removal is rather strange.
Nevertheless, the beta number does not increase under the vertex-removal, same as the independence number and the Lov\'{a}sz number. More explicitly, $\beta(G') \le \beta(G)$ if $G'$ is an induced subgraph of $G$.

Tensor product of systems is used in quantum mechanics often. Denote $G$ the anti-commutativity and commutativity graph corresponding to the tensor product of the SAURs of $G_1$ and $G_2$. We can directly verify that $G$ is the XOR product of $G_1, G_2$, that is, $(i_1,j_1) \sim (i_2, j_2)$ if and only if only one of $i_1\sim i_2$ and $j_1\sim j_2$ holds. In this case, denote $G = G_1 \times G_2$.
\begin{theorem}
  For two given graphs $G_1$ and $G_2$, $\beta(G_1\times G_2) \ge \beta(G_1) \beta(G_2)$.
\end{theorem}
\begin{proof}
  Denote $\{A_i\}$ and $\{B_j\}$ the standard SAUR of $G_1$ and $G_2$, respectively. Then we know that
  \begin{equation}
    \sos{\{A_i\}} = \beta(G_1), \sos{\{B_j\}} = \beta(G_2).
  \end{equation} 
  Hence,
  \begin{align}
    \beta(G_1\times G_2) &\ge \sos{\{A_i\otimes B_j\}}\nonumber\\ 
    &\ge \max_{\rho_1\otimes\rho_2} \sum_{ij} \mean{A_i\otimes B_j}_{\rho_1\otimes\rho_2}^2\nonumber\\
    &=\beta(G_1)\beta(G_2),
  \end{align} 
  since $\{A_i\otimes B_j\}_{ij}$ is a SAUR of $G_1 \times G_2$. 
\end{proof}
For a perfect graph $G$, we know that $\alpha(G) = \vartheta(G) = \alpha^\star(G)$ which implies that $\alpha(G) = \beta(G) = \vartheta(G) = \alpha^\star(G)$. For imperfect graphs, odd cycles and odd anticycles are basic building blocks. To continue, we make the following claim:
\begin{theorem}\label{cl:ladder2}
\begin{equation}\label{eq:ladder}
    \max_\rho\!\left[\mean{\id Y}_{\rho}^2\!+\!\mean{XX}_{\rho}^2\!+\!\mean{ZZ}_{\rho}^2\!-\!\mean{YY}_{\rho}^2\right]\!=\!1.
\end{equation}
\end{theorem}
\begin{proof}
    It is enough to show that the maximum of the following SDP is $1$, as a relaxation of the original theorem:
    \begin{align}
    l = & \max_{\gamma} \tr \gamma W \nonumber\\
    \text{such that } & \tr A_0 \gamma = 1, \gamma \succeq 0,\nonumber\\
                      & \tr A_i\gamma \ge 0, i=1,2,3, \label{eq:extraConds}
\end{align}
where $W = \id Y\id Y + XXXX + ZZZZ - YYYY$, $A_0 = \id\id\id\id$, $A_1 = XZXZ$, $A_2=Y\id Y\id$ and $A_3=ZXZX$. 

Then the dual SDP is
\begin{align}
    l' = & \min y_0\nonumber\\
    \text{such that } & \sum_{i=0}^3 y_i A_i - {W} \succeq 0,\nonumber\\
    & y_i \le 0, i=1,2,3.
\end{align}
Since the case that $y_0 = 1$ and $y_i=-1$ for $i=1,2,3$ is a feasible solution, we know that $l \le l' \le 1$. However, by taking $\rho = |\Psi^+\rangle\langle \Psi^+|$ with $|\Psi^+\rangle = (|00\rangle + |11\rangle)/\sqrt{2}$ in Eq.~\eqref{eq:ladder}, we have $l\ge 1$. Consequently, $l=1$.
\end{proof}
We remark that it is already known that~\cite{luo2008quantum} $l=1$ if we only consider the state in the form $\rho = (\id + x XX + y YY + z ZZ)/4$, for which $\mean{\id Y}_{\rho} = 0$. Hence, our result generalizes the known one, which is considered often in quantum correlations like discord~\cite{luo2008quantum}.

By considering the relation between the standard SAUR of different odd cycles, we have
\begin{lemma}\label{thm:ladder2}
   $\beta(C_{2n+3}) - \beta(C_{2n+1}) \le 1$. 
\end{lemma}
\begin{proof}
    For convenience, we label the vertice of $C_{2k+1}$ in the way such that 
    \begin{equation}
        1\sim 3\sim \cdots \sim 2k+1 \sim 2(k-1) \sim \cdots \sim 2,
    \end{equation}
    where $k=n, n+1$.
    
    The proof is based on the observation that the standard SAUR $\{S'_i\}_{i=1}^{2n+3}$ of $C_{2n+3}$ can be constructed as follows:
    \begin{align}
        &S'_i = S_i \id_2,\ \forall i=1, \ldots, 2n,\nonumber\\
        &S'_{2n+1} = \id_{2^{n-1}}  XX, S'_{2n+2} = \id_{2^{n-1}}  ZZ,\nonumber\\ 
        &S'_{2n+3} = \id_{2^{n}}  Y,
    \end{align}
    where $\{S_i\}_{i=1}^{2n} \cup \{\id_{2^{n-1}} Y\}$ is the standard SAUR of $C_{2n+1}$.

    Denote $S''_{2n+1} = \id_{2^{n-1}}  YY$, we have
    \begin{align}
        \beta(C_{2n+3})
        =& \max_{\sigma'} \left[\sum_{i=1}^{2n} \mean{S'_i}_{\sigma'}^2 + \sum_{i=2n+1}^{2n+3} \mean{S'_i}_{\sigma'}^2\right]\nonumber\\ 
        =& \max_{\sigma'} \left[\left(\sum_{i=1}^{2n} \mean{S'_i}_{\sigma'}^2 + \mean{S''_{2n+1}}_{\sigma'}^2\right) + \right.\nonumber\\ 
        &\left. \qquad \left(\sum_{i=2n+1}^{2n+3} \mean{S'_i}_{\sigma'}^2 -\mean{S''_{2n+1}}_{\sigma'}^2\right)\right]\nonumber\\ 
        \le& \beta(C_{2n+1}) + 1,
    \end{align}
    where the last inequality is from the fact that $\{S'_i\}_{i=1}^{2n} \cup \{S''_{2n+1}\}$ is one SAUR of $C_{2n+1}$ and Theorem~\ref{cl:ladder2}.
\end{proof}
Since $\alpha(C_{2n+1}) = n$, $\beta(C_5) = \alpha(C_5) = 2$, Theorem~\ref{cl:ladder2} and Lemma~\ref{thm:ladder2} lead to the following theorem.
\begin{theorem}
    $\cSOS{\saur(C_m)} = \stab(C_m)$.
\end{theorem}
\begin{proof}
    Since $C_m$ is a perfect graph when $m$ is an even number, the fact that $\stab(C_m) = \th(C_m)$ implies $\cSOS{\saur(C_m)} = \stab(C_m)$.

    When $m=2n+1$ is an odd number, Theorem~\ref{cl:ladder2} and Lemma~\ref{thm:ladder2} result in the fact that $\beta(C_{2n+1}) = \alpha(C_{2n+1})$. Consequently, this implies that $\cSOS{\saur(C_{2n+1})} = \stab(C_{2n+1})$, since the only nontrivial facet of $\stab(C_{2n+1})$ has the norm vector  $(1,\cdots,1)$~\cite{benchetrit2017integer}.
\end{proof}
%%%%%%%%%%%%%%%%%%%%%%%%%%%%%%%%%%%%%%%%%%%%%%%%%%%%%%%%%%%%%%%%%%%%%%%%%%%%%%%%%%%%%%%
\begin{theorem}\label{thm:sara}
    $\SOS{\sara(G)} = \SOS{\saura{G}} = \th(G)$.
\end{theorem}

\begin{proof}
 Notice that $\th(G) \subseteq \SOS{\saura(G)} \subseteq \SOS{\sara{G}}$. Hence, it is sufficient to prove $\SOS{\sara{G}} \subseteq \th(G)$, equivalently, $\sos{\sara(G),w} \le \vartheta(G,w)$ for any non-negative weight vector $w$.

For a given SARA $\{A_i\}$ and a state $\rho$, denote $\omega_i = \tr A_i^2 \rho$. The fact that $A_i^2 \le \id$ leads to $\omega_i \le 1$. Then~\cite{wang2022quantum}
\begin{align}
      \sum_i w_i\mean{A_i}_{\rho}^2 \le \lambda(\mathcal{A}),
\end{align}
where $\mathcal{A}$ is a matrix with $(i,j)$-th element $\sqrt{w_iw_j}\mean{\{A_i,A_j\}}/2$. Thus, the $i$-th diagonal term of $\mathcal{A}$ is $w_i\omega_i$, and $(i,j)$-th element is $0$ if $i\sim j$. Denote $w' = (w_i\omega_i)$. Notice that~\cite{knuth1994sandwich}
  \begin{align}
      \vartheta(G,w') &= \max_{\mathcal{B}} \lambda(\mathcal{B}),\nonumber\\
      \text{ s.t. } & \mathcal{B}_{ij} = 0, \text{ if } i\sim j,\\
                    & \mathcal{B}_{ii} = w'_i,\\
                    & \mathcal{B} \ge 0.
  \end{align}
  Thus, by definition, we have $\lambda(\mathcal{A}) \le \vartheta(G,w')$. Meanwhile, $\vartheta(G,w') \le \vartheta(G,w)$ since $w'_i \le w_i$ for each $i$. 
\end{proof}

\begin{lemma}\label{lm:sar_dec}
  For a given graph $G$ and one SAR $\{A_i\}$ of it, there is a unitary $U$ such that 
  \begin{equation}\label{eq:decSar}
    UA_iU^\dagger = \oplus_{t=1}^T A_i^{(t)},\ [A_i^{(t)}]^2 = [\lambda_{i}^{(t)}]^2 \id_{d_t},
  \end{equation}
  where $\{\lambda_i^{(t)}\}$ are singular values of $A_i$.
  Besides, $\{A_i^{(t)}\}$ is an SAR of $G$ for any $t=1, \ldots, T$.
\end{lemma}
\begin{proof}
  We notice that $[A_i^2, A_j]=0$ and $[A_i^2, A_j^2] = 0$ for any pair $(i,j)$. Hence, there is a unitary $U$ to diagonalize all the $A_i^2$'s simultaneously. By ordering the diagonal terms properly, we have the decomposition
  \begin{equation}
    UA_i^2U^\dagger = \oplus_{t=1}^T [\lambda_{i}^{(t)}]^2 \id_{d_t}.
  \end{equation} 
  Denote $\{\langle u|\}$ the rows of $U$. Then by choosing $|u\rangle, |v\rangle$ such that $A_i^2 |u\rangle = [\lambda_i^{(t)}]^2 |u\rangle$ and $A_i^2|v\rangle = [\lambda_i^{(l)}]^2 |v\rangle$, we have
  \begin{align}
    & \langle u| A_i^2 A_j|v\rangle = [\lambda_i^{t}]^2 \langle u|A_j|v\rangle\nonumber\\
    = & \langle u| A_j A_i^2|v\rangle = [\lambda_i^{l}]^2 \langle u|A_j|v\rangle.
  \end{align}
  Hence, $\langle u|A_j|v\rangle = 0$ whenever $\lambda_i^{(t)} \neq \lambda_i^{(l)}$. This leads to the desired decomposition as in Eq.~\eqref{eq:decSar}.
\end{proof}
\begin{lemma}\label{lm:star}
    For any given $\{S_i\} \in \saur(G)$, state $\rho$ and weight vector $w$ where $|w_i|\le 1$, there exists $\{P_i\} \in \saur(G)$ and state $\tau$ such that $\mean{P_i}_\tau = w_i \mean{S_i}_\rho$.
\end{lemma}
\begin{proof}
    Denote 
    \begin{align}
        P_i = S_i \otimes (\otimes_j Z(w_i)^{\delta_{ij}}),\ \tau = \rho\otimes (|+\rangle\langle +|)^{\otimes n},
    \end{align}
    where $Z(w) = w X + \sqrt{1-w^2} Z$ for $t\in [-1,1]$, $\delta_{ij} = 1$ if $i=j$, otherwise, $\delta_{ij} = 0$, and $n$ is the number of vertices in $G$.
    Direct calculation concludes the proof.
\end{proof}

We remark here that Lemma~\ref{lm:star} also implies that $\SOS{\saur(G)}$ is star-convex.
\begin{theorem}
    $\SOS{\sar(G)} = \SOS{\saur{G}}$.
\end{theorem}
\begin{proof}
    It is equivalent to prove that, for any given $\{S_i\} \in \sar(G)$ and state $\rho$, there is a $\{P_i\} \in \saur(G)$ and state $\tau$ such that $\mean{S_i}_{\rho} = \mean{P_i}_\tau$.
    Without loss of generality, we assume that $S_i = \oplus_{t=1}^T S_i^{(t)}$, where $S_i^{(t)}$ acts on $\mathcal{H}^{(t)}$. A direct calculation shows that
    \begin{equation}
        \mean{S_i}_\rho = \sum_t p_t \mean{S_i^{(t)}}_{\rho_t},
    \end{equation}
    where $\rho_t$ is the block of $\rho$ in $\mathcal{H}^{(t)}$ up to the normalization coefficient $p_t$.

    From Lemma~\ref{lm:star}, we know that there exists $\{P_i^{(t)}\} \in \saur(G)$ and state $\tau_t$ such that $\mean{S_i^{(t)}}_{\rho_t} = \mean{P_i^{(t)}}_{\tau_t}$. Consequently,
    \begin{equation}\label{eq:middleStep}
        \mean{S_i}_\rho = \sum_t p_t \mean{P_i^{(t)}}_{\tau_t}.
    \end{equation}
    Denote $P_i = \oplus_t P_i^{(t)}$ and $\tau = \oplus_t p_t\tau_t$, we have $\{P_i\} \in \saur(G)$ and $\tau$ is a quantum state. Eq.~\eqref{eq:middleStep} implies that $\mean{S_i}_\rho = \mean{P_i}_{\tau}$.
\end{proof}

\twocolumngrid
\bibliography{ref.bib}

\begin{thebibliography}{61}%
\makeatletter
\providecommand \@ifxundefined [1]{%
 \@ifx{#1\undefined}
}%
\providecommand \@ifnum [1]{%
 \ifnum #1\expandafter \@firstoftwo
 \else \expandafter \@secondoftwo
 \fi
}%
\providecommand \@ifx [1]{%
 \ifx #1\expandafter \@firstoftwo
 \else \expandafter \@secondoftwo
 \fi
}%
\providecommand \natexlab [1]{#1}%
\providecommand \enquote  [1]{``#1''}%
\providecommand \bibnamefont  [1]{#1}%
\providecommand \bibfnamefont [1]{#1}%
\providecommand \citenamefont [1]{#1}%
\providecommand \href@noop [0]{\@secondoftwo}%
\providecommand \href [0]{\begingroup \@sanitize@url \@href}%
\providecommand \@href[1]{\@@startlink{#1}\@@href}%
\providecommand \@@href[1]{\endgroup#1\@@endlink}%
\providecommand \@sanitize@url [0]{\catcode `\\12\catcode `\$12\catcode
  `\&12\catcode `\#12\catcode `\^12\catcode `\_12\catcode `\%12\relax}%
\providecommand \@@startlink[1]{}%
\providecommand \@@endlink[0]{}%
\providecommand \url  [0]{\begingroup\@sanitize@url \@url }%
\providecommand \@url [1]{\endgroup\@href {#1}{\urlprefix }}%
\providecommand \urlprefix  [0]{URL }%
\providecommand \Eprint [0]{\href }%
\providecommand \doibase [0]{https://doi.org/}%
\providecommand \selectlanguage [0]{\@gobble}%
\providecommand \bibinfo  [0]{\@secondoftwo}%
\providecommand \bibfield  [0]{\@secondoftwo}%
\providecommand \translation [1]{[#1]}%
\providecommand \BibitemOpen [0]{}%
\providecommand \bibitemStop [0]{}%
\providecommand \bibitemNoStop [0]{.\EOS\space}%
\providecommand \EOS [0]{\spacefactor3000\relax}%
\providecommand \BibitemShut  [1]{\csname bibitem#1\endcsname}%
\let\auto@bib@innerbib\@empty
%</preamble>
\bibitem [{\citenamefont {Kaye}\ \emph {et~al.}(2006)\citenamefont {Kaye},
  \citenamefont {Laflamme},\ and\ \citenamefont
  {Mosca}}]{kaye2006introduction}%
  \BibitemOpen
  \bibfield  {author} {\bibinfo {author} {\bibfnamefont {P.}~\bibnamefont
  {Kaye}}, \bibinfo {author} {\bibfnamefont {R.}~\bibnamefont {Laflamme}},\
  and\ \bibinfo {author} {\bibfnamefont {M.}~\bibnamefont {Mosca}},\ }\href
  {https://doi.org/10.1093/oso/9780198570004.001.0001} {\emph {\bibinfo {title}
  {An Introduction to Quantum Computing}}}\ (\bibinfo  {publisher} {Oxford
  University Press},\ \bibinfo {year} {2006})\BibitemShut {NoStop}%
\bibitem [{\citenamefont {Stolze}\ and\ \citenamefont
  {Suter}(2008)}]{stolze2008quantum}%
  \BibitemOpen
  \bibfield  {author} {\bibinfo {author} {\bibfnamefont {J.}~\bibnamefont
  {Stolze}}\ and\ \bibinfo {author} {\bibfnamefont {D.}~\bibnamefont {Suter}},\
  }\href {https://doi.org/10.1002/9783527617760} {\emph {\bibinfo {title}
  {Quantum computing: a short course from theory to experiment}}}\ (\bibinfo
  {publisher} {John Wiley {\&} Sons},\ \bibinfo {year} {2008})\BibitemShut
  {NoStop}%
\bibitem [{\citenamefont {Briegel}\ \emph {et~al.}(2009)\citenamefont
  {Briegel}, \citenamefont {Browne}, \citenamefont {D{\"u}r}, \citenamefont
  {Raussendorf},\ and\ \citenamefont {Van~den Nest}}]{briegel2009measurement}%
  \BibitemOpen
  \bibfield  {author} {\bibinfo {author} {\bibfnamefont {H.~J.}\ \bibnamefont
  {Briegel}}, \bibinfo {author} {\bibfnamefont {D.~E.}\ \bibnamefont {Browne}},
  \bibinfo {author} {\bibfnamefont {W.}~\bibnamefont {D{\"u}r}}, \bibinfo
  {author} {\bibfnamefont {R.}~\bibnamefont {Raussendorf}},\ and\ \bibinfo
  {author} {\bibfnamefont {M.}~\bibnamefont {Van~den Nest}},\ }\bibfield
  {title} {\bibinfo {title} {Measurement-based quantum computation},\ }\href
  {https://doi.org/10.1038/nphys1157} {\bibfield  {journal} {\bibinfo
  {journal} {Nature Phys.}\ }\textbf {\bibinfo {volume} {5}},\ \bibinfo {pages}
  {19} (\bibinfo {year} {2009})}\BibitemShut {NoStop}%
\bibitem [{\citenamefont {Cafaro}\ and\ \citenamefont
  {Mancini}(2011)}]{cafaro2011geometric}%
  \BibitemOpen
  \bibfield  {author} {\bibinfo {author} {\bibfnamefont {C.}~\bibnamefont
  {Cafaro}}\ and\ \bibinfo {author} {\bibfnamefont {S.}~\bibnamefont
  {Mancini}},\ }\bibfield  {title} {\bibinfo {title} {A geometric algebra
  perspective on quantum computational gates and universality in quantum
  computing},\ }\href {https://doi.org/10.1007/s00006-010-0269-x} {\bibfield
  {journal} {\bibinfo  {journal} {Adv. Appl. Clifford Algebras}\ }\textbf
  {\bibinfo {volume} {21}},\ \bibinfo {pages} {493} (\bibinfo {year}
  {2011})}\BibitemShut {NoStop}%
\bibitem [{\citenamefont {Bharti}\ \emph {et~al.}(2022)\citenamefont {Bharti},
  \citenamefont {Cervera-Lierta}, \citenamefont {Kyaw}, \citenamefont {Haug},
  \citenamefont {Alperin-Lea}, \citenamefont {Anand}, \citenamefont {Degroote},
  \citenamefont {Heimonen}, \citenamefont {Kottmann}, \citenamefont {Menke},
  \citenamefont {Mok}, \citenamefont {Sim}, \citenamefont {Kwek},\ and\
  \citenamefont {Aspuru-Guzik}}]{bharti2022noisy}%
  \BibitemOpen
  \bibfield  {author} {\bibinfo {author} {\bibfnamefont {K.}~\bibnamefont
  {Bharti}}, \bibinfo {author} {\bibfnamefont {A.}~\bibnamefont
  {Cervera-Lierta}}, \bibinfo {author} {\bibfnamefont {T.~H.}\ \bibnamefont
  {Kyaw}}, \bibinfo {author} {\bibfnamefont {T.}~\bibnamefont {Haug}}, \bibinfo
  {author} {\bibfnamefont {S.}~\bibnamefont {Alperin-Lea}}, \bibinfo {author}
  {\bibfnamefont {A.}~\bibnamefont {Anand}}, \bibinfo {author} {\bibfnamefont
  {M.}~\bibnamefont {Degroote}}, \bibinfo {author} {\bibfnamefont
  {H.}~\bibnamefont {Heimonen}}, \bibinfo {author} {\bibfnamefont {J.~S.}\
  \bibnamefont {Kottmann}}, \bibinfo {author} {\bibfnamefont {T.}~\bibnamefont
  {Menke}}, \bibinfo {author} {\bibfnamefont {W.-K.}\ \bibnamefont {Mok}},
  \bibinfo {author} {\bibfnamefont {S.}~\bibnamefont {Sim}}, \bibinfo {author}
  {\bibfnamefont {L.-C.}\ \bibnamefont {Kwek}},\ and\ \bibinfo {author}
  {\bibfnamefont {A.}~\bibnamefont {Aspuru-Guzik}},\ }\bibfield  {title}
  {\bibinfo {title} {Noisy intermediate-scale quantum algorithms},\ }\href
  {https://doi.org/10.1103/RevModPhys.94.015004} {\bibfield  {journal}
  {\bibinfo  {journal} {Rev. Mod. Phys.}\ }\textbf {\bibinfo {volume} {94}},\
  \bibinfo {pages} {015004} (\bibinfo {year} {2022})}\BibitemShut {NoStop}%
\bibitem [{\citenamefont {Bennett}\ and\ \citenamefont
  {Brassard}(1984)}]{bennett2020quantum}%
  \BibitemOpen
  \bibfield  {author} {\bibinfo {author} {\bibfnamefont {C.~H.}\ \bibnamefont
  {Bennett}}\ and\ \bibinfo {author} {\bibfnamefont {G.}~\bibnamefont
  {Brassard}},\ }\bibfield  {title} {\bibinfo {title} {{Quantum cryptography:
  Public key distribution and coin tossing}},\ }in\ \href
  {https://doi.org/10.48550/arXiv.2003.06557} {\emph {\bibinfo {booktitle}
  {Proc. IEEE International Conference on Computers, Systems, and Signal
  Processing, Bangalore, 1-012 Dec 1984}}}\ (\bibinfo {year} {1984})\ pp.\
  \bibinfo {pages} {175--179}\BibitemShut {NoStop}%
\bibitem [{\citenamefont {Hahn}\ \emph {et~al.}(2019)\citenamefont {Hahn},
  \citenamefont {Pappa},\ and\ \citenamefont {Eisert}}]{hahn2019quantum}%
  \BibitemOpen
  \bibfield  {author} {\bibinfo {author} {\bibfnamefont {F.}~\bibnamefont
  {Hahn}}, \bibinfo {author} {\bibfnamefont {A.}~\bibnamefont {Pappa}},\ and\
  \bibinfo {author} {\bibfnamefont {J.}~\bibnamefont {Eisert}},\ }\bibfield
  {title} {\bibinfo {title} {Quantum network routing and local
  complementation},\ }\href {https://doi.org/10.1038/s41534-019-0191-6}
  {\bibfield  {journal} {\bibinfo  {journal} {npj Quantum Inf.}\ }\textbf
  {\bibinfo {volume} {5}},\ \bibinfo {pages} {76} (\bibinfo {year}
  {2019})}\BibitemShut {NoStop}%
\bibitem [{\citenamefont {Nema}\ and\ \citenamefont
  {Nene}(2020)}]{nema2020pauli}%
  \BibitemOpen
  \bibfield  {author} {\bibinfo {author} {\bibfnamefont {P.}~\bibnamefont
  {Nema}}\ and\ \bibinfo {author} {\bibfnamefont {M.~J.}\ \bibnamefont
  {Nene}},\ }\bibfield  {title} {\bibinfo {title} {Pauli matrix based quantum
  communication protocol},\ }in\ \href
  {https://doi.org/10.1109/ICATMRI51801.2020.9398393} {\emph {\bibinfo
  {booktitle} {2020 IEEE International Conference on Advent Trends in
  Multidisciplinary Research and Innovation (ICATMRI)}}}\ (\bibinfo
  {organization} {IEEE},\ \bibinfo {year} {2020})\ pp.\ \bibinfo {pages}
  {1--6}\BibitemShut {NoStop}%
\bibitem [{\citenamefont {Mastriani}\ \emph {et~al.}(2021)\citenamefont
  {Mastriani}, \citenamefont {Iyengar},\ and\ \citenamefont
  {Kumar}}]{mastriani2021satellite}%
  \BibitemOpen
  \bibfield  {author} {\bibinfo {author} {\bibfnamefont {M.}~\bibnamefont
  {Mastriani}}, \bibinfo {author} {\bibfnamefont {S.~S.}\ \bibnamefont
  {Iyengar}},\ and\ \bibinfo {author} {\bibfnamefont {L.}~\bibnamefont
  {Kumar}},\ }\bibfield  {title} {\bibinfo {title} {Satellite quantum
  communication protocol regardless of the weather},\ }\href
  {https://doi.org/10.1007/s11082-021-02829-8} {\bibfield  {journal} {\bibinfo
  {journal} {Opt. Quantum. Electron.}\ }\textbf {\bibinfo {volume} {53}},\
  \bibinfo {pages} {181} (\bibinfo {year} {2021})}\BibitemShut {NoStop}%
\bibitem [{\citenamefont {Luttinger}\ and\ \citenamefont
  {Ward}(1960)}]{luttinger1960ground}%
  \BibitemOpen
  \bibfield  {author} {\bibinfo {author} {\bibfnamefont {J.~M.}\ \bibnamefont
  {Luttinger}}\ and\ \bibinfo {author} {\bibfnamefont {J.~C.}\ \bibnamefont
  {Ward}},\ }\bibfield  {title} {\bibinfo {title} {Ground-state energy of a
  many-fermion system {II}},\ }\href {https://doi.org/10.1103/PhysRev.118.1417}
  {\bibfield  {journal} {\bibinfo  {journal} {Phys. Rev.}\ }\textbf {\bibinfo
  {volume} {118}},\ \bibinfo {pages} {1417} (\bibinfo {year}
  {1960})}\BibitemShut {NoStop}%
\bibitem [{\citenamefont {Chapman}\ \emph {et~al.}(2023)\citenamefont
  {Chapman}, \citenamefont {Elman},\ and\ \citenamefont
  {Mann}}]{chapman2023unified}%
  \BibitemOpen
  \bibfield  {author} {\bibinfo {author} {\bibfnamefont {A.}~\bibnamefont
  {Chapman}}, \bibinfo {author} {\bibfnamefont {S.~J.}\ \bibnamefont {Elman}},\
  and\ \bibinfo {author} {\bibfnamefont {R.~L.}\ \bibnamefont {Mann}},\ }\href
  {https://doi.org/10.48550/arXiv.2305.15625} {\bibinfo {title} {{A Unified
  Graph-Theoretic Framework for Free-Fermion Solvability}}} (\bibinfo {year}
  {2023}),\ \bibinfo {note} {arXiv[quant-ph]:2305.15625}\BibitemShut {NoStop}%
\bibitem [{\citenamefont {Chapman}\ and\ \citenamefont
  {Flammia}(2020)}]{chapman2020characterization}%
  \BibitemOpen
  \bibfield  {author} {\bibinfo {author} {\bibfnamefont {A.}~\bibnamefont
  {Chapman}}\ and\ \bibinfo {author} {\bibfnamefont {S.~T.}\ \bibnamefont
  {Flammia}},\ }\bibfield  {title} {\bibinfo {title} {Characterization of
  solvable spin models via graph invariants},\ }\href
  {https://doi.org/10.22331/q-2020-06-04-278} {\bibfield  {journal} {\bibinfo
  {journal} {Quantum}\ }\textbf {\bibinfo {volume} {4}},\ \bibinfo {pages}
  {278} (\bibinfo {year} {2020})}\BibitemShut {NoStop}%
\bibitem [{\citenamefont {Kitaev}(2015)}]{kitaev2015simple}%
  \BibitemOpen
  \bibfield  {author} {\bibinfo {author} {\bibfnamefont {A.}~\bibnamefont
  {Kitaev}},\ }\href {https://online.kitp.ucsb.edu/online/entangled15/kitaev/}
  {\bibinfo {title} {A simple model of quantum holography}} (\bibinfo {year}
  {2015}),\ \bibinfo {note} {{KITP} Program: Entanglement in
  Strongly-Correlated Quantum Matter}\BibitemShut {NoStop}%
\bibitem [{\citenamefont {Sachdev}\ and\ \citenamefont
  {Ye}(1993)}]{sachdev1993gapless}%
  \BibitemOpen
  \bibfield  {author} {\bibinfo {author} {\bibfnamefont {S.}~\bibnamefont
  {Sachdev}}\ and\ \bibinfo {author} {\bibfnamefont {J.}~\bibnamefont {Ye}},\
  }\bibfield  {title} {\bibinfo {title} {Gapless spin-fluid ground state in a
  random quantum {Heisenberg} magnet},\ }\href
  {https://doi.org/10.1103/PhysRevLett.70.3339} {\bibfield  {journal} {\bibinfo
   {journal} {Phys. Rev. Lett.}\ }\textbf {\bibinfo {volume} {70}},\ \bibinfo
  {pages} {3339} (\bibinfo {year} {1993})}\BibitemShut {NoStop}%
\bibitem [{\citenamefont {Pirandola}\ \emph {et~al.}(2020)\citenamefont
  {Pirandola}, \citenamefont {Andersen}, \citenamefont {Banchi}, \citenamefont
  {Berta}, \citenamefont {Bunandar}, \citenamefont {Colbeck}, \citenamefont
  {Englund}, \citenamefont {Gehring}, \citenamefont {Lupo}, \citenamefont
  {Ottaviani} \emph {et~al.}}]{pirandola2020advances}%
  \BibitemOpen
  \bibfield  {author} {\bibinfo {author} {\bibfnamefont {S.}~\bibnamefont
  {Pirandola}}, \bibinfo {author} {\bibfnamefont {U.~L.}\ \bibnamefont
  {Andersen}}, \bibinfo {author} {\bibfnamefont {L.}~\bibnamefont {Banchi}},
  \bibinfo {author} {\bibfnamefont {M.}~\bibnamefont {Berta}}, \bibinfo
  {author} {\bibfnamefont {D.}~\bibnamefont {Bunandar}}, \bibinfo {author}
  {\bibfnamefont {R.}~\bibnamefont {Colbeck}}, \bibinfo {author} {\bibfnamefont
  {D.}~\bibnamefont {Englund}}, \bibinfo {author} {\bibfnamefont
  {T.}~\bibnamefont {Gehring}}, \bibinfo {author} {\bibfnamefont
  {C.}~\bibnamefont {Lupo}}, \bibinfo {author} {\bibfnamefont {C.}~\bibnamefont
  {Ottaviani}}, \emph {et~al.},\ }\bibfield  {title} {\bibinfo {title}
  {Advances in quantum cryptography},\ }\href
  {https://doi.org/10.1364/AOP.361502} {\bibfield  {journal} {\bibinfo
  {journal} {Adv. Opt. Photonics}\ }\textbf {\bibinfo {volume} {12}},\ \bibinfo
  {pages} {1012} (\bibinfo {year} {2020})}\BibitemShut {NoStop}%
\bibitem [{\citenamefont {Fannes}\ \emph {et~al.}(1992)\citenamefont {Fannes},
  \citenamefont {Nachtergaele},\ and\ \citenamefont
  {Werner}}]{fannes1992finitely}%
  \BibitemOpen
  \bibfield  {author} {\bibinfo {author} {\bibfnamefont {M.}~\bibnamefont
  {Fannes}}, \bibinfo {author} {\bibfnamefont {B.}~\bibnamefont
  {Nachtergaele}},\ and\ \bibinfo {author} {\bibfnamefont {R.~F.}\ \bibnamefont
  {Werner}},\ }\bibfield  {title} {\bibinfo {title} {Finitely correlated states
  on quantum spin chains},\ }\href {https://doi.org/10.1007/BF02099178}
  {\bibfield  {journal} {\bibinfo  {journal} {Commun. Math. Phys.}\ }\textbf
  {\bibinfo {volume} {144}},\ \bibinfo {pages} {443} (\bibinfo {year}
  {1992})}\BibitemShut {NoStop}%
\bibitem [{\citenamefont {Halmos}(2012)}]{halmos2012hilbert}%
  \BibitemOpen
  \bibfield  {author} {\bibinfo {author} {\bibfnamefont {P.~R.}\ \bibnamefont
  {Halmos}},\ }\href {https://doi.org/10.1007/978-1-4684-9330-6} {\emph
  {\bibinfo {title} {A Hilbert Space Problem Pook}}},\ Vol.~\bibinfo {volume}
  {19}\ (\bibinfo  {publisher} {Springer Science \& Business Media},\ \bibinfo
  {year} {2012})\BibitemShut {NoStop}%
\bibitem [{\citenamefont {Gawron}\ \emph {et~al.}(2010)\citenamefont {Gawron},
  \citenamefont {Pucha{\l}a}, \citenamefont {Miszczak}, \citenamefont
  {Skowronek},\ and\ \citenamefont {{\.Z}yczkowski}}]{gawron2010restricted}%
  \BibitemOpen
  \bibfield  {author} {\bibinfo {author} {\bibfnamefont {P.}~\bibnamefont
  {Gawron}}, \bibinfo {author} {\bibfnamefont {Z.}~\bibnamefont {Pucha{\l}a}},
  \bibinfo {author} {\bibfnamefont {J.~A.}\ \bibnamefont {Miszczak}}, \bibinfo
  {author} {\bibfnamefont {{\L}.}~\bibnamefont {Skowronek}},\ and\ \bibinfo
  {author} {\bibfnamefont {K.}~\bibnamefont {{\.Z}yczkowski}},\ }\bibfield
  {title} {\bibinfo {title} {Restricted numerical range: a versatile tool in
  the theory of quantum information},\ }\href
  {https://doi.org/10.1063/1.3496901} {\bibfield  {journal} {\bibinfo
  {journal} {J. Math. Phys.}\ }\textbf {\bibinfo {volume} {51}},\ \bibinfo
  {pages} {10} (\bibinfo {year} {2010})}\BibitemShut {NoStop}%
\bibitem [{\citenamefont {Pucha{\l}a}\ \emph {et~al.}(2011)\citenamefont
  {Pucha{\l}a}, \citenamefont {Gawron}, \citenamefont {Miszczak}, \citenamefont
  {Skowronek}, \citenamefont {Choi},\ and\ \citenamefont
  {{\.Z}yczkowski}}]{puchala2011product}%
  \BibitemOpen
  \bibfield  {author} {\bibinfo {author} {\bibfnamefont {Z.}~\bibnamefont
  {Pucha{\l}a}}, \bibinfo {author} {\bibfnamefont {P.}~\bibnamefont {Gawron}},
  \bibinfo {author} {\bibfnamefont {J.~A.}\ \bibnamefont {Miszczak}}, \bibinfo
  {author} {\bibfnamefont {{\L}.}~\bibnamefont {Skowronek}}, \bibinfo {author}
  {\bibfnamefont {M.-D.}\ \bibnamefont {Choi}},\ and\ \bibinfo {author}
  {\bibfnamefont {K.}~\bibnamefont {{\.Z}yczkowski}},\ }\bibfield  {title}
  {\bibinfo {title} {Product numerical range in a space with tensor product
  structure},\ }\href {https://doi.org/10.1016/j.laa.2010.08.026} {\bibfield
  {journal} {\bibinfo  {journal} {Lin. Alg. Appl.}\ }\textbf {\bibinfo {volume}
  {434}},\ \bibinfo {pages} {327} (\bibinfo {year} {2011})}\BibitemShut
  {NoStop}%
\bibitem [{\citenamefont {Dunkl}\ \emph {et~al.}(2011)\citenamefont {Dunkl},
  \citenamefont {Gawron}, \citenamefont {Holbrook}, \citenamefont
  {Pucha{\l}a},\ and\ \citenamefont {{\.Z}yczkowski}}]{dunkl2011numerical}%
  \BibitemOpen
  \bibfield  {author} {\bibinfo {author} {\bibfnamefont {C.~F.}\ \bibnamefont
  {Dunkl}}, \bibinfo {author} {\bibfnamefont {P.}~\bibnamefont {Gawron}},
  \bibinfo {author} {\bibfnamefont {J.~A.}\ \bibnamefont {Holbrook}}, \bibinfo
  {author} {\bibfnamefont {Z.}~\bibnamefont {Pucha{\l}a}},\ and\ \bibinfo
  {author} {\bibfnamefont {K.}~\bibnamefont {{\.Z}yczkowski}},\ }\bibfield
  {title} {\bibinfo {title} {Numerical shadows: {M}easures and densities on the
  numerical range},\ }\href {https://doi.org/10.1016/j.laa.2010.12.003}
  {\bibfield  {journal} {\bibinfo  {journal} {Lin. Alg. Appl.}\ }\textbf
  {\bibinfo {volume} {434}},\ \bibinfo {pages} {2042} (\bibinfo {year}
  {2011})}\BibitemShut {NoStop}%
\bibitem [{\citenamefont {Gutkin}\ and\ \citenamefont
  {{\.Z}yczkowski}(2013)}]{gutkin2013joint}%
  \BibitemOpen
  \bibfield  {author} {\bibinfo {author} {\bibfnamefont {E.}~\bibnamefont
  {Gutkin}}\ and\ \bibinfo {author} {\bibfnamefont {K.}~\bibnamefont
  {{\.Z}yczkowski}},\ }\bibfield  {title} {\bibinfo {title} {Joint numerical
  ranges, quantum maps, and joint numerical shadows},\ }\href
  {https://doi.org/10.1016/j.laa.2012.10.043} {\bibfield  {journal} {\bibinfo
  {journal} {Lin. Alg. Appl.}\ }\textbf {\bibinfo {volume} {438}},\ \bibinfo
  {pages} {2394} (\bibinfo {year} {2013})}\BibitemShut {NoStop}%
\bibitem [{\citenamefont {Szyma{\'n}ski}\ \emph {et~al.}(2018)\citenamefont
  {Szyma{\'n}ski}, \citenamefont {Weis},\ and\ \citenamefont
  {{\.Z}yczkowski}}]{szymanski2018classification}%
  \BibitemOpen
  \bibfield  {author} {\bibinfo {author} {\bibfnamefont {K.}~\bibnamefont
  {Szyma{\'n}ski}}, \bibinfo {author} {\bibfnamefont {S.}~\bibnamefont
  {Weis}},\ and\ \bibinfo {author} {\bibfnamefont {K.}~\bibnamefont
  {{\.Z}yczkowski}},\ }\bibfield  {title} {\bibinfo {title} {Classification of
  joint numerical ranges of three {Hermitian} matrices of size three},\ }\href
  {https://doi.org/10.1016/j.laa.2017.11.017} {\bibfield  {journal} {\bibinfo
  {journal} {Lin. Alg. Appl.}\ }\textbf {\bibinfo {volume} {545}},\ \bibinfo
  {pages} {148} (\bibinfo {year} {2018})}\BibitemShut {NoStop}%
\bibitem [{\citenamefont {Szyma{\'n}ski}\ and\ \citenamefont
  {{\.Z}yczkowski}(2019)}]{szymanski2019geometric}%
  \BibitemOpen
  \bibfield  {author} {\bibinfo {author} {\bibfnamefont {K.}~\bibnamefont
  {Szyma{\'n}ski}}\ and\ \bibinfo {author} {\bibfnamefont {K.}~\bibnamefont
  {{\.Z}yczkowski}},\ }\bibfield  {title} {\bibinfo {title} {Geometric and
  algebraic origins of additive uncertainty relations},\ }\href
  {https://doi.org/10.1088/1751-8121/ab4543} {\bibfield  {journal} {\bibinfo
  {journal} {J. Phys. A: Math. Theor.}\ }\textbf {\bibinfo {volume} {53}},\
  \bibinfo {pages} {015302} (\bibinfo {year} {2019})}\BibitemShut {NoStop}%
\bibitem [{\citenamefont {Perez-Garcia}\ \emph {et~al.}(2007)\citenamefont
  {Perez-Garcia}, \citenamefont {Verstraete}, \citenamefont {Wolf},\ and\
  \citenamefont {Cirac}}]{perez2007matrix}%
  \BibitemOpen
  \bibfield  {author} {\bibinfo {author} {\bibfnamefont {D.}~\bibnamefont
  {Perez-Garcia}}, \bibinfo {author} {\bibfnamefont {F.}~\bibnamefont
  {Verstraete}}, \bibinfo {author} {\bibfnamefont {M.~M.}\ \bibnamefont
  {Wolf}},\ and\ \bibinfo {author} {\bibfnamefont {J.~I.}\ \bibnamefont
  {Cirac}},\ }\bibfield  {title} {\bibinfo {title} {{Matrix Product State
  Representations}},\ }\href {https://doi.org/10.26421/QIC7.5-6-1} {\bibfield
  {journal} {\bibinfo  {journal} {Quantum Info. Comput.}\ }\textbf {\bibinfo
  {volume} {7}},\ \bibinfo {pages} {401–430} (\bibinfo {year}
  {2007})}\BibitemShut {NoStop}%
\bibitem [{\citenamefont {Schuch}\ \emph {et~al.}(2008)\citenamefont {Schuch},
  \citenamefont {Cirac},\ and\ \citenamefont
  {Verstraete}}]{schuch2008computational}%
  \BibitemOpen
  \bibfield  {author} {\bibinfo {author} {\bibfnamefont {N.}~\bibnamefont
  {Schuch}}, \bibinfo {author} {\bibfnamefont {I.}~\bibnamefont {Cirac}},\ and\
  \bibinfo {author} {\bibfnamefont {F.}~\bibnamefont {Verstraete}},\ }\bibfield
   {title} {\bibinfo {title} {Computational difficulty of finding matrix
  product ground states},\ }\href
  {https://doi.org/10.1103/PhysRevLett.100.250501} {\bibfield  {journal}
  {\bibinfo  {journal} {Phys. Rev. Lett.}\ }\textbf {\bibinfo {volume} {100}},\
  \bibinfo {pages} {250501} (\bibinfo {year} {2008})}\BibitemShut {NoStop}%
\bibitem [{\citenamefont {Hadfield}\ \emph {et~al.}(2019)\citenamefont
  {Hadfield}, \citenamefont {Wang}, \citenamefont {O’gorman}, \citenamefont
  {Rieffel}, \citenamefont {Venturelli},\ and\ \citenamefont
  {Biswas}}]{hadfield2019quantum}%
  \BibitemOpen
  \bibfield  {author} {\bibinfo {author} {\bibfnamefont {S.}~\bibnamefont
  {Hadfield}}, \bibinfo {author} {\bibfnamefont {Z.}~\bibnamefont {Wang}},
  \bibinfo {author} {\bibfnamefont {B.}~\bibnamefont {O’gorman}}, \bibinfo
  {author} {\bibfnamefont {E.~G.}\ \bibnamefont {Rieffel}}, \bibinfo {author}
  {\bibfnamefont {D.}~\bibnamefont {Venturelli}},\ and\ \bibinfo {author}
  {\bibfnamefont {R.}~\bibnamefont {Biswas}},\ }\bibfield  {title} {\bibinfo
  {title} {From the quantum approximate optimization algorithm to a quantum
  alternating operator ansatz},\ }\href {https://doi.org/10.3390/a12020034}
  {\bibfield  {journal} {\bibinfo  {journal} {Algorithms}\ }\textbf {\bibinfo
  {volume} {12}},\ \bibinfo {pages} {34} (\bibinfo {year} {2019})}\BibitemShut
  {NoStop}%
\bibitem [{\citenamefont {Farhi}\ \emph {et~al.}(2014)\citenamefont {Farhi},
  \citenamefont {Goldstone},\ and\ \citenamefont {Gutmann}}]{farhi2014quantum}%
  \BibitemOpen
  \bibfield  {author} {\bibinfo {author} {\bibfnamefont {E.}~\bibnamefont
  {Farhi}}, \bibinfo {author} {\bibfnamefont {J.}~\bibnamefont {Goldstone}},\
  and\ \bibinfo {author} {\bibfnamefont {S.}~\bibnamefont {Gutmann}},\ }\href
  {https://doi.org/10.48550/arXiv.1411.4028} {\bibinfo {title} {A quantum
  approximate optimization algorithm}} (\bibinfo {year} {2014}),\ \bibinfo
  {note} {arXiv[quant-ph]:1411.4028}\BibitemShut {NoStop}%
\bibitem [{\citenamefont {Binkowski}\ \emph {et~al.}(2023)\citenamefont
  {Binkowski}, \citenamefont {Ko{\ss}mann}, \citenamefont {Ziegler},\ and\
  \citenamefont {Schwonnek}}]{binkowski2023elementary}%
  \BibitemOpen
  \bibfield  {author} {\bibinfo {author} {\bibfnamefont {L.}~\bibnamefont
  {Binkowski}}, \bibinfo {author} {\bibfnamefont {G.}~\bibnamefont
  {Ko{\ss}mann}}, \bibinfo {author} {\bibfnamefont {T.}~\bibnamefont
  {Ziegler}},\ and\ \bibinfo {author} {\bibfnamefont {R.}~\bibnamefont
  {Schwonnek}},\ }\href {https://doi.org/10.48550/arXiv.2302.04968} {\bibinfo
  {title} {{Elementary Proof of QAOA Convergence}}} (\bibinfo {year} {2023}),\
  \bibinfo {note} {arXiv[quant-ph]:2302.04968}\BibitemShut {NoStop}%
\bibitem [{\citenamefont {Li}\ and\ \citenamefont {Lu}(2023)}]{li2023quantum}%
  \BibitemOpen
  \bibfield  {author} {\bibinfo {author} {\bibfnamefont {B.}~\bibnamefont
  {Li}}\ and\ \bibinfo {author} {\bibfnamefont {J.}~\bibnamefont {Lu}},\ }\href
  {https://doi.org/10.48550/arXiv.2305.18571} {\bibinfo {title} {{Quantum
  Variational Embedding for Ground-State Energy Problems: Sum of Squares and
  Cluster Selection}}} (\bibinfo {year} {2023}),\ \bibinfo {note}
  {arXiv[quant-ph]:2305.18571}\BibitemShut {NoStop}%
\bibitem [{\citenamefont {Kandala}\ \emph {et~al.}(2017)\citenamefont
  {Kandala}, \citenamefont {Mezzacapo}, \citenamefont {Temme}, \citenamefont
  {Takita}, \citenamefont {Brink}, \citenamefont {Chow},\ and\ \citenamefont
  {Gambetta}}]{kandala2017hardware}%
  \BibitemOpen
  \bibfield  {author} {\bibinfo {author} {\bibfnamefont {A.}~\bibnamefont
  {Kandala}}, \bibinfo {author} {\bibfnamefont {A.}~\bibnamefont {Mezzacapo}},
  \bibinfo {author} {\bibfnamefont {K.}~\bibnamefont {Temme}}, \bibinfo
  {author} {\bibfnamefont {M.}~\bibnamefont {Takita}}, \bibinfo {author}
  {\bibfnamefont {M.}~\bibnamefont {Brink}}, \bibinfo {author} {\bibfnamefont
  {J.~M.}\ \bibnamefont {Chow}},\ and\ \bibinfo {author} {\bibfnamefont
  {J.~M.}\ \bibnamefont {Gambetta}},\ }\bibfield  {title} {\bibinfo {title}
  {Hardware-efficient variational quantum eigensolver for small molecules and
  quantum magnets},\ }\href {https://doi.org/10.1038/nature23879} {\bibfield
  {journal} {\bibinfo  {journal} {Nature}\ }\textbf {\bibinfo {volume} {549}},\
  \bibinfo {pages} {242} (\bibinfo {year} {2017})}\BibitemShut {NoStop}%
\bibitem [{\citenamefont {Hastings}\ and\ \citenamefont
  {O'Donnell}(2022)}]{hastings2022optimizing}%
  \BibitemOpen
  \bibfield  {author} {\bibinfo {author} {\bibfnamefont {M.~B.}\ \bibnamefont
  {Hastings}}\ and\ \bibinfo {author} {\bibfnamefont {R.}~\bibnamefont
  {O'Donnell}},\ }\bibfield  {title} {\bibinfo {title} {Optimizing strongly
  interacting fermionic {Hamiltonians}},\ }in\ \href
  {https://doi.org/10.1145/3519935.3519960} {\emph {\bibinfo {booktitle} {Proc.
  54th Annual ACM SIGACT Symposium on Theory of Computing}}}\ (\bibinfo {year}
  {2022})\ pp.\ \bibinfo {pages} {776--789}\BibitemShut {NoStop}%
\bibitem [{\citenamefont {Kull}\ \emph {et~al.}(2022)\citenamefont {Kull},
  \citenamefont {Schuch}, \citenamefont {Dive},\ and\ \citenamefont
  {Navascu{\'e}s}}]{kull2022lower}%
  \BibitemOpen
  \bibfield  {author} {\bibinfo {author} {\bibfnamefont {I.}~\bibnamefont
  {Kull}}, \bibinfo {author} {\bibfnamefont {N.}~\bibnamefont {Schuch}},
  \bibinfo {author} {\bibfnamefont {B.}~\bibnamefont {Dive}},\ and\ \bibinfo
  {author} {\bibfnamefont {M.}~\bibnamefont {Navascu{\'e}s}},\ }\href
  {https://doi.org/10.48550/arXiv.2212.03014} {\bibinfo {title} {{Lower
  Bounding Ground-State Energies of Local Hamiltonians Through the
  Renormalization Group}}} (\bibinfo {year} {2022}),\ \bibinfo {note}
  {arXiv[quant-ph]:2212.03014}\BibitemShut {NoStop}%
\bibitem [{\citenamefont {Eisert}(2023)}]{eisert2023note}%
  \BibitemOpen
  \bibfield  {author} {\bibinfo {author} {\bibfnamefont {J.}~\bibnamefont
  {Eisert}},\ }\href {https://doi.org/10.48550/arXiv.2301.06142} {\bibinfo
  {title} {A note on lower bounds to variational problems with guarantees}}
  (\bibinfo {year} {2023}),\ \bibinfo {note}
  {arXiv[quant-ph]:2301.06142}\BibitemShut {NoStop}%
\bibitem [{\citenamefont {Ligthart}\ \emph {et~al.}(2023)\citenamefont
  {Ligthart}, \citenamefont {Gachechiladze},\ and\ \citenamefont
  {Gross}}]{ligthart2023convergent}%
  \BibitemOpen
  \bibfield  {author} {\bibinfo {author} {\bibfnamefont {L.~T.}\ \bibnamefont
  {Ligthart}}, \bibinfo {author} {\bibfnamefont {M.}~\bibnamefont
  {Gachechiladze}},\ and\ \bibinfo {author} {\bibfnamefont {D.}~\bibnamefont
  {Gross}},\ }\bibfield  {title} {\bibinfo {title} {A convergent inflation
  hierarchy for quantum causal structures},\ }\href
  {https://doi.org/10.1007/s00220-023-04697-7} {\bibfield  {journal} {\bibinfo
  {journal} {Commun. Math. Phys.}\ }\textbf {\bibinfo {volume} {402}},\
  \bibinfo {pages} {1} (\bibinfo {year} {2023})}\BibitemShut {NoStop}%
\bibitem [{\citenamefont {Pironio}\ \emph {et~al.}(2010)\citenamefont
  {Pironio}, \citenamefont {Navascu{\'e}s},\ and\ \citenamefont
  {Ac{\'\i}n}}]{pironio2010convergent}%
  \BibitemOpen
  \bibfield  {author} {\bibinfo {author} {\bibfnamefont {S.}~\bibnamefont
  {Pironio}}, \bibinfo {author} {\bibfnamefont {M.}~\bibnamefont
  {Navascu{\'e}s}},\ and\ \bibinfo {author} {\bibfnamefont {A.}~\bibnamefont
  {Ac{\'\i}n}},\ }\bibfield  {title} {\bibinfo {title} {Convergent relaxations
  of polynomial optimization problems with noncommuting variables},\ }\href
  {https://doi.org/10.1137/090760155} {\bibfield  {journal} {\bibinfo
  {journal} {SIAM J. Optim.}\ }\textbf {\bibinfo {volume} {20}},\ \bibinfo
  {pages} {2157} (\bibinfo {year} {2010})}\BibitemShut {NoStop}%
\bibitem [{\citenamefont {Kellner}\ \emph {et~al.}(2013)\citenamefont
  {Kellner}, \citenamefont {Theobald},\ and\ \citenamefont
  {Trabandt}}]{kellner2013containment}%
  \BibitemOpen
  \bibfield  {author} {\bibinfo {author} {\bibfnamefont {K.}~\bibnamefont
  {Kellner}}, \bibinfo {author} {\bibfnamefont {T.}~\bibnamefont {Theobald}},\
  and\ \bibinfo {author} {\bibfnamefont {C.}~\bibnamefont {Trabandt}},\
  }\bibfield  {title} {\bibinfo {title} {Containment problems for polytopes and
  spectrahedra},\ }\href {https://doi.org/10.1137/120874898} {\bibfield
  {journal} {\bibinfo  {journal} {SIAM J. Optim.}\ }\textbf {\bibinfo {volume}
  {23}},\ \bibinfo {pages} {1000} (\bibinfo {year} {2013})}\BibitemShut
  {NoStop}%
\bibitem [{\citenamefont {Kellner}\ \emph {et~al.}(2015)\citenamefont
  {Kellner}, \citenamefont {Theobald},\ and\ \citenamefont
  {Trabandt}}]{kellner2015semidefinite}%
  \BibitemOpen
  \bibfield  {author} {\bibinfo {author} {\bibfnamefont {K.}~\bibnamefont
  {Kellner}}, \bibinfo {author} {\bibfnamefont {T.}~\bibnamefont {Theobald}},\
  and\ \bibinfo {author} {\bibfnamefont {C.}~\bibnamefont {Trabandt}},\
  }\bibfield  {title} {\bibinfo {title} {A semidefinite hierarchy for
  containment of spectrahedra},\ }\href {https://doi.org/10.1137/140971634}
  {\bibfield  {journal} {\bibinfo  {journal} {SIAM J. Optim.}\ }\textbf
  {\bibinfo {volume} {25}},\ \bibinfo {pages} {1013} (\bibinfo {year}
  {2015})}\BibitemShut {NoStop}%
\bibitem [{\citenamefont {Gokhale}\ \emph {et~al.}(2019)\citenamefont
  {Gokhale}, \citenamefont {Angiuli}, \citenamefont {Ding}, \citenamefont
  {Gui}, \citenamefont {Tomesh}, \citenamefont {Suchara}, \citenamefont
  {Martonosi},\ and\ \citenamefont {Chong}}]{gokhale2019minimizing}%
  \BibitemOpen
  \bibfield  {author} {\bibinfo {author} {\bibfnamefont {P.}~\bibnamefont
  {Gokhale}}, \bibinfo {author} {\bibfnamefont {O.}~\bibnamefont {Angiuli}},
  \bibinfo {author} {\bibfnamefont {Y.}~\bibnamefont {Ding}}, \bibinfo {author}
  {\bibfnamefont {K.}~\bibnamefont {Gui}}, \bibinfo {author} {\bibfnamefont
  {T.}~\bibnamefont {Tomesh}}, \bibinfo {author} {\bibfnamefont
  {M.}~\bibnamefont {Suchara}}, \bibinfo {author} {\bibfnamefont
  {M.}~\bibnamefont {Martonosi}},\ and\ \bibinfo {author} {\bibfnamefont
  {F.~T.}\ \bibnamefont {Chong}},\ }\href
  {https://doi.org/10.48550/arXiv.1907.13623} {\bibinfo {title} {Minimizing
  state preparations in variational quantum eigensolver by partitioning into
  commuting families}} (\bibinfo {year} {2019}),\ \bibinfo {note}
  {arXiv[quant-ph]:1907.13623}\BibitemShut {NoStop}%
\bibitem [{\citenamefont {Kirby}\ and\ \citenamefont
  {Love}(2019)}]{kirby2019contextuality}%
  \BibitemOpen
  \bibfield  {author} {\bibinfo {author} {\bibfnamefont {W.~M.}\ \bibnamefont
  {Kirby}}\ and\ \bibinfo {author} {\bibfnamefont {P.~J.}\ \bibnamefont
  {Love}},\ }\bibfield  {title} {\bibinfo {title} {Contextuality test of the
  nonclassicality of variational quantum eigensolvers},\ }\href
  {https://doi.org/10.1103/PhysRevLett.123.200501} {\bibfield  {journal}
  {\bibinfo  {journal} {Phys. Rev. Lett.}\ }\textbf {\bibinfo {volume} {123}},\
  \bibinfo {pages} {200501} (\bibinfo {year} {2019})}\BibitemShut {NoStop}%
\bibitem [{\citenamefont {de~Gois}\ \emph {et~al.}(2023)\citenamefont
  {de~Gois}, \citenamefont {Hansenne},\ and\ \citenamefont
  {G{\"u}hne}}]{de2022uncertainty}%
  \BibitemOpen
  \bibfield  {author} {\bibinfo {author} {\bibfnamefont {C.}~\bibnamefont
  {de~Gois}}, \bibinfo {author} {\bibfnamefont {K.}~\bibnamefont {Hansenne}},\
  and\ \bibinfo {author} {\bibfnamefont {O.}~\bibnamefont {G{\"u}hne}},\
  }\bibfield  {title} {\bibinfo {title} {Uncertainty relations from graph
  theory},\ }\href {https://doi.org/10.1103/PhysRevA.107.062211} {\bibfield
  {journal} {\bibinfo  {journal} {Phys. Rev. A}\ }\textbf {\bibinfo {volume}
  {107}},\ \bibinfo {pages} {062211} (\bibinfo {year} {2023})}\BibitemShut
  {NoStop}%
\bibitem [{\citenamefont {Kaniewski}\ \emph {et~al.}(2014)\citenamefont
  {Kaniewski}, \citenamefont {Tomamichel},\ and\ \citenamefont
  {Wehner}}]{kaniewski2014entropic}%
  \BibitemOpen
  \bibfield  {author} {\bibinfo {author} {\bibfnamefont {J.}~\bibnamefont
  {Kaniewski}}, \bibinfo {author} {\bibfnamefont {M.}~\bibnamefont
  {Tomamichel}},\ and\ \bibinfo {author} {\bibfnamefont {S.}~\bibnamefont
  {Wehner}},\ }\bibfield  {title} {\bibinfo {title} {Entropic uncertainty from
  effective anticommutators},\ }\href
  {https://doi.org/10.1103/PhysRevA.90.012332} {\bibfield  {journal} {\bibinfo
  {journal} {Phys. Rev. A}\ }\textbf {\bibinfo {volume} {90}},\ \bibinfo
  {pages} {012332} (\bibinfo {year} {2014})}\BibitemShut {NoStop}%
\bibitem [{\citenamefont {Gastineau-Hills}(1982)}]{gastineau1982quasi}%
  \BibitemOpen
  \bibfield  {author} {\bibinfo {author} {\bibfnamefont {H.}~\bibnamefont
  {Gastineau-Hills}},\ }\bibfield  {title} {\bibinfo {title} {Quasi clifford
  algebras and systems of orthogonal designs},\ }\href
  {https://doi.org/10.1017/S1446788700024368} {\bibfield  {journal} {\bibinfo
  {journal} {J. Austral. Math. Soc.}\ }\textbf {\bibinfo {volume} {32}},\
  \bibinfo {pages} {1} (\bibinfo {year} {1982})}\BibitemShut {NoStop}%
\bibitem [{\citenamefont {Gastineau-Hills}(1981)}]{gastineau1981systems}%
  \BibitemOpen
  \bibfield  {author} {\bibinfo {author} {\bibfnamefont {H.~M.}\ \bibnamefont
  {Gastineau-Hills}},\ }\bibfield  {title} {\bibinfo {title} {Systems of
  orthogonal designs and quasi clifford algebras},\ }\href
  {https://doi.org/10.1017/S000497270000753X} {\bibfield  {journal} {\bibinfo
  {journal} {Bull. Austral. Math. Soc.}\ }\textbf {\bibinfo {volume} {24}},\
  \bibinfo {pages} {157} (\bibinfo {year} {1981})}\BibitemShut {NoStop}%
\bibitem [{\citenamefont {Knuth}(1994)}]{knuth1994sandwich}%
  \BibitemOpen
  \bibfield  {author} {\bibinfo {author} {\bibfnamefont {D.~E.}\ \bibnamefont
  {Knuth}},\ }\bibfield  {title} {\bibinfo {title} {The sandwich theorem},\
  }\href {https://doi.org/10.37236/1193} {\bibfield  {journal} {\bibinfo
  {journal} {Electron. J. Comb.}\ }\textbf {\bibinfo {volume} {1}},\ \bibinfo
  {pages} {A1} (\bibinfo {year} {1994})}\BibitemShut {NoStop}%
\bibitem [{\citenamefont {Budroni}\ \emph {et~al.}(2022)\citenamefont
  {Budroni}, \citenamefont {Cabello}, \citenamefont {G{\"u}hne}, \citenamefont
  {Kleinmann},\ and\ \citenamefont {Larsson}}]{budroni2022kochen}%
  \BibitemOpen
  \bibfield  {author} {\bibinfo {author} {\bibfnamefont {C.}~\bibnamefont
  {Budroni}}, \bibinfo {author} {\bibfnamefont {A.}~\bibnamefont {Cabello}},
  \bibinfo {author} {\bibfnamefont {O.}~\bibnamefont {G{\"u}hne}}, \bibinfo
  {author} {\bibfnamefont {M.}~\bibnamefont {Kleinmann}},\ and\ \bibinfo
  {author} {\bibfnamefont {J.-{\AA}.}\ \bibnamefont {Larsson}},\ }\bibfield
  {title} {\bibinfo {title} {{Kochen-Specker} contextuality},\ }\href
  {https://doi.org/10.1103/RevModPhys.94.045007} {\bibfield  {journal}
  {\bibinfo  {journal} {Rev. Mod. Phys.}\ }\textbf {\bibinfo {volume} {94}},\
  \bibinfo {pages} {045007} (\bibinfo {year} {2022})}\BibitemShut {NoStop}%
\bibitem [{\citenamefont {Cabello}\ \emph {et~al.}(2014)\citenamefont
  {Cabello}, \citenamefont {Severini},\ and\ \citenamefont
  {Winter}}]{cabello2014graph}%
  \BibitemOpen
  \bibfield  {author} {\bibinfo {author} {\bibfnamefont {A.}~\bibnamefont
  {Cabello}}, \bibinfo {author} {\bibfnamefont {S.}~\bibnamefont {Severini}},\
  and\ \bibinfo {author} {\bibfnamefont {A.}~\bibnamefont {Winter}},\
  }\bibfield  {title} {\bibinfo {title} {Graph-theoretic approach to quantum
  correlations},\ }\href {https://doi.org/10.1103/PhysRevLett.112.040401}
  {\bibfield  {journal} {\bibinfo  {journal} {Phys. Rev. Lett.}\ }\textbf
  {\bibinfo {volume} {112}},\ \bibinfo {pages} {040401} (\bibinfo {year}
  {2014})}\BibitemShut {NoStop}%
\bibitem [{\citenamefont {Samoilenko}(1991)}]{samoilenko2012spectral}%
  \BibitemOpen
  \bibfield  {author} {\bibinfo {author} {\bibfnamefont {Y.~S.}\ \bibnamefont
  {Samoilenko}},\ }\href {https://doi.org/10.1007/978-94-011-3806-2} {\emph
  {\bibinfo {title} {Spectral Theory of Families of Self-Adjoint Operators}}},\
  \bibinfo {series} {Mathematics and its Applications}, Vol.~\bibinfo {volume}
  {57}\ (\bibinfo  {publisher} {Springer Science \& Business Media},\ \bibinfo
  {year} {1991})\BibitemShut {NoStop}%
\bibitem [{\citenamefont {Geller}\ and\ \citenamefont
  {Stahl}(1975)}]{geller1975chromatic}%
  \BibitemOpen
  \bibfield  {author} {\bibinfo {author} {\bibfnamefont {D.}~\bibnamefont
  {Geller}}\ and\ \bibinfo {author} {\bibfnamefont {S.}~\bibnamefont {Stahl}},\
  }\bibfield  {title} {\bibinfo {title} {The chromatic number and other
  functions of the lexicographic product},\ }\href
  {https://doi.org/10.1016/0095-8956(75)90076-3} {\bibfield  {journal}
  {\bibinfo  {journal} {J. Comb. Theory. Ser. B}\ }\textbf {\bibinfo {volume}
  {19}},\ \bibinfo {pages} {87} (\bibinfo {year} {1975})}\BibitemShut {NoStop}%
\bibitem [{\citenamefont {Lov{\'a}sz}(1979)}]{lovasz1979shannon}%
  \BibitemOpen
  \bibfield  {author} {\bibinfo {author} {\bibfnamefont {L.}~\bibnamefont
  {Lov{\'a}sz}},\ }\bibfield  {title} {\bibinfo {title} {On the {Shannon}
  capacity of a graph},\ }\href {https://doi.org/10.1109/TIT.1979.1055985}
  {\bibfield  {journal} {\bibinfo  {journal} {IEEE Trans. Inf. Theory}\
  }\textbf {\bibinfo {volume} {25}},\ \bibinfo {pages} {1} (\bibinfo {year}
  {1979})}\BibitemShut {NoStop}%
\bibitem [{\citenamefont {Cubitt}\ \emph {et~al.}(2014)\citenamefont {Cubitt},
  \citenamefont {Man{\v{c}}inska}, \citenamefont {Roberson}, \citenamefont
  {Severini}, \citenamefont {Stahlke},\ and\ \citenamefont
  {Winter}}]{cubitt2014bounds}%
  \BibitemOpen
  \bibfield  {author} {\bibinfo {author} {\bibfnamefont {T.}~\bibnamefont
  {Cubitt}}, \bibinfo {author} {\bibfnamefont {L.}~\bibnamefont
  {Man{\v{c}}inska}}, \bibinfo {author} {\bibfnamefont {D.~E.}\ \bibnamefont
  {Roberson}}, \bibinfo {author} {\bibfnamefont {S.}~\bibnamefont {Severini}},
  \bibinfo {author} {\bibfnamefont {D.}~\bibnamefont {Stahlke}},\ and\ \bibinfo
  {author} {\bibfnamefont {A.}~\bibnamefont {Winter}},\ }\bibfield  {title}
  {\bibinfo {title} {Bounds on entanglement-assisted source-channel coding via
  the {Lov{\'a}sz} $\vartheta$-number and its variants},\ }\href
  {https://doi.org/10.1109/TIT.2014.2349502} {\bibfield  {journal} {\bibinfo
  {journal} {IEEE Trans. Inf. Theory}\ }\textbf {\bibinfo {volume} {60}},\
  \bibinfo {pages} {7330} (\bibinfo {year} {2014})}\BibitemShut {NoStop}%
\bibitem [{\citenamefont {Hastings}(2024)}]{hastings2024limitations}%
  \BibitemOpen
  \bibfield  {author} {\bibinfo {author} {\bibfnamefont {M.~B.}\ \bibnamefont
  {Hastings}},\ }\href {https://doi.org/10.48550/arXiv.2402.14752} {\bibinfo
  {title} {Limitations and separations in the quantum sum-of-squares, and the
  quantum knapsack problem}} (\bibinfo {year} {2024}),\ \bibinfo {note}
  {arXiv[quant-ph]:2402.14752}\BibitemShut {NoStop}%
\bibitem [{\citenamefont {de~Guise}\ \emph {et~al.}(2018)\citenamefont
  {de~Guise}, \citenamefont {Maccone}, \citenamefont {Sanders},\ and\
  \citenamefont {Shukla}}]{de2018state}%
  \BibitemOpen
  \bibfield  {author} {\bibinfo {author} {\bibfnamefont {H.}~\bibnamefont
  {de~Guise}}, \bibinfo {author} {\bibfnamefont {L.}~\bibnamefont {Maccone}},
  \bibinfo {author} {\bibfnamefont {B.~C.}\ \bibnamefont {Sanders}},\ and\
  \bibinfo {author} {\bibfnamefont {N.}~\bibnamefont {Shukla}},\ }\bibfield
  {title} {\bibinfo {title} {State-independent uncertainty relations},\ }\href
  {https://doi.org/10.1103/PhysRevA.98.042121} {\bibfield  {journal} {\bibinfo
  {journal} {Phys. Rev. A}\ }\textbf {\bibinfo {volume} {98}},\ \bibinfo
  {pages} {042121} (\bibinfo {year} {2018})}\BibitemShut {NoStop}%
\bibitem [{\citenamefont {Fritz}(2021)}]{fritz2021unified}%
  \BibitemOpen
  \bibfield  {author} {\bibinfo {author} {\bibfnamefont {T.}~\bibnamefont
  {Fritz}},\ }\bibfield  {title} {\bibinfo {title} {A unified construction of
  semiring-homomorphic graph invariants},\ }\href
  {https://doi.org/10.1007/s10801-020-00983-y} {\bibfield  {journal} {\bibinfo
  {journal} {J. Algebraic Comb.}\ }\textbf {\bibinfo {volume} {54}},\ \bibinfo
  {pages} {693} (\bibinfo {year} {2021})}\BibitemShut {NoStop}%
\bibitem [{\citenamefont {Hofmann}\ and\ \citenamefont
  {Takeuchi}(2003)}]{hofmann2003violation}%
  \BibitemOpen
  \bibfield  {author} {\bibinfo {author} {\bibfnamefont {H.~F.}\ \bibnamefont
  {Hofmann}}\ and\ \bibinfo {author} {\bibfnamefont {S.}~\bibnamefont
  {Takeuchi}},\ }\bibfield  {title} {\bibinfo {title} {Violation of local
  uncertainty relations as a signature of entanglement},\ }\href
  {https://doi.org/10.1103/PhysRevA.68.032103} {\bibfield  {journal} {\bibinfo
  {journal} {Phys. Rev. A}\ }\textbf {\bibinfo {volume} {68}},\ \bibinfo
  {pages} {032103} (\bibinfo {year} {2003})}\BibitemShut {NoStop}%
\bibitem [{\citenamefont {Wehner}\ and\ \citenamefont
  {Winter}(2008)}]{wehner2008higher}%
  \BibitemOpen
  \bibfield  {author} {\bibinfo {author} {\bibfnamefont {S.}~\bibnamefont
  {Wehner}}\ and\ \bibinfo {author} {\bibfnamefont {A.}~\bibnamefont
  {Winter}},\ }\bibfield  {title} {\bibinfo {title} {Higher entropic
  uncertainty relations for anti-commuting observables},\ }\href
  {https://doi.org/10.1063/1.2943685} {\bibfield  {journal} {\bibinfo
  {journal} {J. Math. Phys.}\ }\textbf {\bibinfo {volume} {49}},\ \bibinfo
  {pages} {062105} (\bibinfo {year} {2008})}\BibitemShut {NoStop}%
\bibitem [{\citenamefont {Karp}(1972)}]{karp1972reducibility}%
  \BibitemOpen
  \bibfield  {author} {\bibinfo {author} {\bibfnamefont {R.~M.}\ \bibnamefont
  {Karp}},\ }\bibfield  {title} {\bibinfo {title} {Reducibility among
  combinatorial problems},\ }in\ \href
  {https://doi.org/10.1007/978-1-4684-2001-2_9} {\emph {\bibinfo {booktitle}
  {Complexity of computer computations}}}\ (\bibinfo  {publisher} {Springer},\
  \bibinfo {year} {1972})\ pp.\ \bibinfo {pages} {85--103}\BibitemShut
  {NoStop}%
\bibitem [{\citenamefont {Yu}\ \emph {et~al.}(2022)\citenamefont {Yu},
  \citenamefont {Simnacher}, \citenamefont {Nguyen},\ and\ \citenamefont
  {G{\"u}hne}}]{yu2022quantum}%
  \BibitemOpen
  \bibfield  {author} {\bibinfo {author} {\bibfnamefont {X.-D.}\ \bibnamefont
  {Yu}}, \bibinfo {author} {\bibfnamefont {T.}~\bibnamefont {Simnacher}},
  \bibinfo {author} {\bibfnamefont {H.~C.}\ \bibnamefont {Nguyen}},\ and\
  \bibinfo {author} {\bibfnamefont {O.}~\bibnamefont {G{\"u}hne}},\ }\bibfield
  {title} {\bibinfo {title} {Quantum-inspired hierarchy for rank-constrained
  optimization},\ }\href {https://doi.org/10.1103/PRXQuantum.3.010340}
  {\bibfield  {journal} {\bibinfo  {journal} {PRX Quantum}\ }\textbf {\bibinfo
  {volume} {3}},\ \bibinfo {pages} {010340} (\bibinfo {year}
  {2022})}\BibitemShut {NoStop}%
\bibitem [{\citenamefont {Wyderka}\ and\ \citenamefont
  {G{\"u}hne}(2020)}]{wyderka2020characterizing}%
  \BibitemOpen
  \bibfield  {author} {\bibinfo {author} {\bibfnamefont {N.}~\bibnamefont
  {Wyderka}}\ and\ \bibinfo {author} {\bibfnamefont {O.}~\bibnamefont
  {G{\"u}hne}},\ }\bibfield  {title} {\bibinfo {title} {Characterizing quantum
  states via sector lengths},\ }\href
  {https://doi.org/10.1088/1751-8121/ab7f0a} {\bibfield  {journal} {\bibinfo
  {journal} {J. Phys. A: Math. Theor.}\ }\textbf {\bibinfo {volume} {53}},\
  \bibinfo {pages} {345302} (\bibinfo {year} {2020})}\BibitemShut {NoStop}%
\bibitem [{\citenamefont {Luo}(2008)}]{luo2008quantum}%
  \BibitemOpen
  \bibfield  {author} {\bibinfo {author} {\bibfnamefont {S.}~\bibnamefont
  {Luo}},\ }\bibfield  {title} {\bibinfo {title} {Quantum discord for two-qubit
  systems},\ }\href {https://doi.org/10.1103/PhysRevA.77.042303} {\bibfield
  {journal} {\bibinfo  {journal} {Phys. Rev. A}\ }\textbf {\bibinfo {volume}
  {77}},\ \bibinfo {pages} {042303} (\bibinfo {year} {2008})}\BibitemShut
  {NoStop}%
\bibitem [{\citenamefont {Benchetrit}(2017)}]{benchetrit2017integer}%
  \BibitemOpen
  \bibfield  {author} {\bibinfo {author} {\bibfnamefont {Y.}~\bibnamefont
  {Benchetrit}},\ }\bibfield  {title} {\bibinfo {title} {Integer round-up
  property for the chromatic number of some h-perfect graphs},\ }\href
  {https://doi.org/10.1007/s10107-016-1085-4} {\bibfield  {journal} {\bibinfo
  {journal} {Math. Program.}\ }\textbf {\bibinfo {volume} {164}},\ \bibinfo
  {pages} {245} (\bibinfo {year} {2017})}\BibitemShut {NoStop}%
\bibitem [{\citenamefont {Wang}\ \emph {et~al.}(2024)\citenamefont {Wang},
  \citenamefont {Xu},\ and\ \citenamefont {G{\"u}hne}}]{wang2022quantum}%
  \BibitemOpen
  \bibfield  {author} {\bibinfo {author} {\bibfnamefont {Y.-X.}\ \bibnamefont
  {Wang}}, \bibinfo {author} {\bibfnamefont {Z.-P.}\ \bibnamefont {Xu}},\ and\
  \bibinfo {author} {\bibfnamefont {O.}~\bibnamefont {G{\"u}hne}},\ }\bibfield
  {title} {\bibinfo {title} {Quantum {LOSR} networks cannot generate graph
  states with high fidelity},\ }\href
  {https://doi.org/10.1038/s41534-024-00806-z} {\bibfield  {journal} {\bibinfo
  {journal} {npj Quantum Inf.}\ }\textbf {\bibinfo {volume} {10}},\ \bibinfo
  {pages} {11} (\bibinfo {year} {2024})}\BibitemShut {NoStop}%
\end{thebibliography}%
\end{document}